\documentclass[a4paper,USenglish,cleveref, autoref, nolineno]{socg-lipics-v2021}

\pdfoutput=1 
\hideLIPIcs  

\usepackage{cite}

\newcommand{\Real}{\ensuremath{\mathbb{R}}}
\newcommand{\Plane}{\ensuremath{\mathbb{R}^2}}

\newcommand{\FVD}{\ensuremath{\mathsf{FVD}}}
\newcommand{\FCVD}{\ensuremath{\mathsf{FCVD}}}

\newcommand{\CVD}{\ensuremath{\mathsf{CVD}}}
\newcommand{\mCVD}{\ensuremath{\overline{\mathsf{CVD}}}}

\newcommand{\HVD}{\ensuremath{\mathsf{HVD}}}

\newcommand{\VD}{\ensuremath{\mathsf{VD}}}

\newcommand{\VR}{\ensuremath{R}}
\newcommand{\mVR}{\ensuremath{\overline{R}}}

\newcommand{\conf}{\mathcal{F}}

\DeclareMathOperator*{\E}{\mathbf{E}}

\newcommand{\arr}{\mathcal{A}}

\newcommand{\dual}[1]{\ensuremath{#1^\star}}
\newcommand{\lift}[1]{\ensuremath{#1^\cup}}

\newcommand{\changed}[1]{{#1}}
\newcommand{\added}[1]{{#1}}

\let\geq\geqslant
\let\leq\leqslant

\newcommand{\Sfplus}{\ensuremath{{S}^{+}_{f}}}
\newcommand{\Sfplusm}{\ensuremath{{S}^{+}_{f'}}}
\newcommand{\Jfplus}{\ensuremath{{J}^{+}_{f}}}

\title{Higher-Order Color Voronoi Diagrams and the Colorful Clarkson--Shor Framework} 

\titlerunning{Higher-Order Color Voronoi Diagrams and the Colorful Clarkson--Shor Framework} 

\author{{Sang Won} Bae}
{Division of AI Computer Science and Engineering, Kyonggi University, Suwon, Republic of Korea}{swbae@kgu.ac.kr}
{https://orcid.org/0000-0002-8802-4247}
{Supported by the National Research Foundation of Korea(NRF) grant funded by the Korea government(MSIT) (No. RS-2023-00251168),
and in part by 
the University of Bayreuth Centre of International Excellence ``Alexander von Humboldt''
(Senior Fellowship 2023).}

\author{Nicolau Oliver}
{Faculty of Informatics, Universit\`{a} della Svizzera italiana, Lugano, Switzerland}
{nicolau.oliver.burwitz@usi.ch}
{https://orcid.org/0009-0004-8901-451X}
{}

\author{Evanthia Papadopoulou}
{Faculty of Informatics, Universit\`{a} della Svizzera italiana, Lugano, Switzerland}
{evanthia.papadopoulou@usi.ch}
{https://orcid.org/0000-0003-0144-7384}
{Supported in part by the Swiss National Science Foundation (SNF), 
project 200021E-201356.}

\authorrunning{S.W. Bae, N. Oliver, and E. Papadopoulou} 

\Copyright{{Sang Won} Bae and Nicolau Oliver and Evanthia Papadopoulou} 



\category{} 

\acknowledgements{
The authors would like to thank Otfried Cheong, Christian Knauer, and Fabian Stehn 
for valuable discussions and comments,
in particular, at the beginning of this work.
The research by the first author has been done
partly during his visits to Universit\"{a}t Bayreuth, Bayreuth, Germany
and to Universit\`{a} della Svizzera italiana, Lugano, Switzerland.
}

\nolinenumbers 

\EventEditors{Oswin Aichholzer and Haitao Wang}
\EventNoEds{2}
\EventLongTitle{41st International Symposium on Computational Geometry (SoCG 2025)}
\EventShortTitle{SoCG 2025}
\EventAcronym{SoCG}
\EventYear{2025}
\EventDate{June 23--27, 2025}
\EventLocation{Kanazawa, Japan}
\EventLogo{socg-logo.pdf}
\SeriesVolume{332}
\ArticleNo{XX}     

\begin{document}

\maketitle

\begin{abstract}
Given a set~$S$ of $n$~colored sites, each $s\in S$ associated with
a distance-to-site function~$\delta_s \colon \Real^2 \to \Real$,
we consider two distance-to-color functions for each color:
one takes the minimum of $\delta_s$ for sites~$s\in S$ in that color
and the other takes the maximum.
These two sets of distance functions induce two families of 
higher-order Voronoi diagrams for colors in the plane,
namely, the minimal and maximal order-$k$ color Voronoi diagrams,
which include various well-studied Voronoi diagrams as special cases.
In this paper,
we derive an exact upper bound $4k(n-k)-2n$
on the total number of vertices in
both the minimal and maximal order-$k$ color diagrams
for a wide class of distance functions~$\delta_s$ that satisfy certain conditions,
including the case \changed{of} 
point sites~$S$ under
\changed{convex distance functions and 
the $L_p$~metric for any~$1\leq p \leq\infty$.}
\changed{For the $L_1$ (or, $L_\infty$) metric,  and other convex
  polygonal metrics,}
we show that the order-$k$ minimal diagram \changed{of point sites} has $O(\min\{k(n-k), (n-k)^2\})$ complexity,
while its maximal \changed{counterpart} 
has $O(\min\{k(n-k), k^2\})$ complexity.
To obtain these combinatorial results,
we extend the Clarkson--Shor framework to colored objects,
and demonstrate its \changed{application} 
to several fundamental geometric structures,
including higher-order color Voronoi diagrams,
colored $j$-facets, and
levels in the arrangements of piecewise linear/algebraic curves/surfaces.
We also present an iterative approach to compute higher-order color Voronoi diagrams.
\end{abstract}

\section{Introduction} \label{sec:intro}

Let $S$ be a set of $n$~\emph{sites},
each of which is assigned a color from a set $K=\{1,\ldots, m\}$ of
$m$~colors. 
Let $S_i\subseteq S$ be the set of sites in color~$i\in K$.
We consider two \emph{distance-to-color} functions from any point $x\in\Plane$ to each color~$i \in K$:
\[ d_i(x) := \min_{s\in S_i} \delta_s(x) \quad \text{ and } \quad
  \bar{d}_i(x) := \max_{s\in S_i} \delta_s(x),\]
where $\delta_s(x)$ denotes the prescribed distance-to-site function from $x\in \Plane$ to
site~$s\in S$.

Based on the two sets of point-to-color distances,
we can define two different Voronoi diagrams of $m$~colors in~$K$.
For each~$1\leq k \leq m-1$ and subset~$H \subseteq K$ with~$|H|=k$,
define
\begin{align*}
 \VR_k(H; S) & := \{ x\in \Plane \mid d_i(x) < d_j(x)
    \text{ for all $i \in H$ and $j \in K\setminus H$} \} \\
 \mVR_k(H; S) & := \{ x\in \Plane \mid \bar{d}_i(x) > \bar{d}_j(x)
    \text{ for all $i \in H$ and $j \in K\setminus H$} \},
\end{align*}
called the \emph{minimal} and \emph{maximal color Voronoi regions} of~$H$
with respect to~$S$.
We then define
the \emph{order-$k$ minimal color Voronoi diagram} of~$S$, $\CVD_k(S)$,
to be the partition of~$\Plane$
into the minimal regions~$\VR_k(H; S)$ for all~$H \subset K$ with $|H| = k$;
the \emph{order-$k$ maximal color Voronoi diagram} of~$S$, $\mCVD_k(S)$,
to be the partition of~$\Plane$
into the maximal regions~$\mVR_k(H; S)$.
In other words, $\CVD_k(S)$ partitions~$\Plane$ 
by $k$~\emph{nearest colors} under the minimal distances~$\{d_i\}_{i\in K}$,
while $\mCVD_k(S)$ partitions~$\Plane$
by $k$~\emph{farthest colors} under the maximal distances~$\{\bar{d}_i\}_{i\in K}$.

\begin{figure}[t]
\begin{center}
\includegraphics[width=\textwidth]{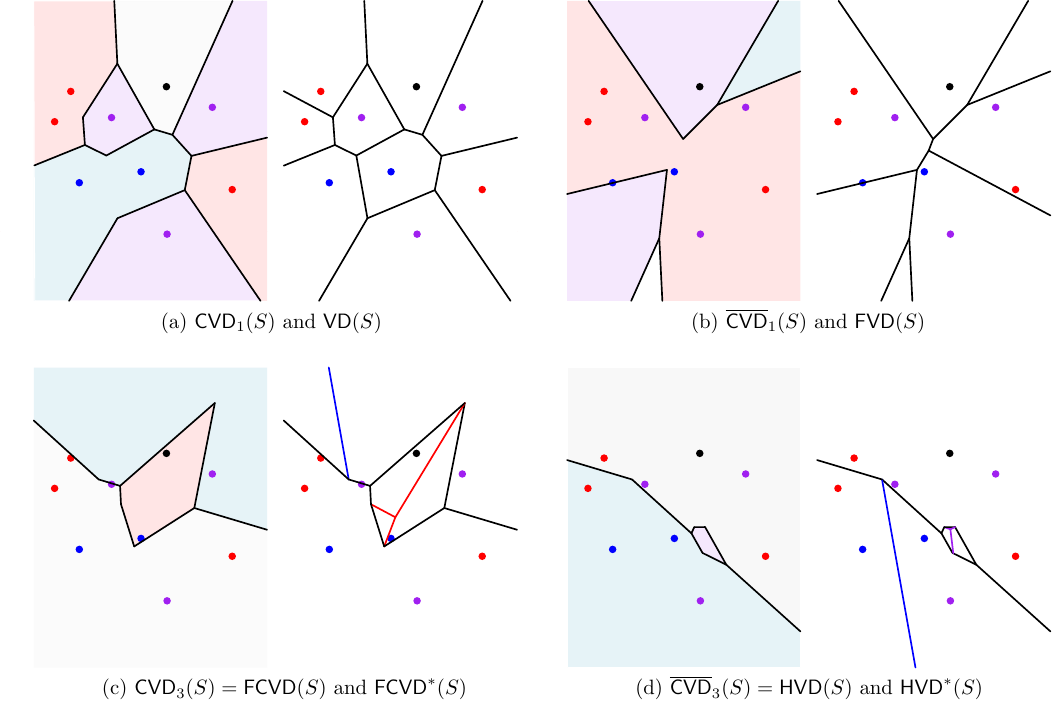}
\end{center}
\caption{
 Special cases of $\CVD_k(S)$ and $\mCVD_k(S)$
 for 
 colored points~$S$ 
 in the Euclidean plane.
}
\label{fig:cvd_special}
\end{figure}

\changed{These two families of color Voronoi diagrams generalize various
conventional counterparts.}
\begin{itemize} 
\item
 For $k=1$, 
 $\CVD_1(S)$ and $\mCVD_1(S)$ correspond to the ordinary nearest-site
 and farthest-site Voronoi diagrams of~$S$, $\VD(S)$ and $\FVD(S)$, respectively,
 where adjacent faces of a common color are merged.
 See \figurename~\ref{fig:cvd_special}(a) and~(b).
 \item If $m=n$, that is, each site in~$S$ has a distinct color,
 then $\CVD_k(S) = \VD_k(S)$, the ordinary order-$k$ Voronoi diagram 
 without colors~\cite{l-knnvdp-82,lpl-knnvdr-15,bcklpz-choavd-15,pz-hovdls-16}.
 In this case, it holds that $\mCVD_k(S)=\VD_{n-k}(S)$,
 also known as the \emph{order-$k$ farthest-site Voronoi diagram} or
 the \emph{order-$k$ maximal Voronoi diagram}~\cite{a-pd:paa-87}.
 \item 
 The \emph{farthest color Voronoi diagram}~$\FCVD(S)$
 \cite{hks-uevsia-93,ahiklmps-fcvdrp-06,mpss-fcvd:ca-20,b-tbiafcvdls-14} 
 partitions the plane~$\Plane$
 into regions of colors~$i\in K$ that consist of all points~$p\in \Plane$
 whose \emph{farthest} color is~$i$ with respect to the minimal distances~$\{d_i\}_{i\in K}$.
 The \emph{Hausdorff Voronoi diagram}~$\HVD(S)$,
 also known as the \emph{min-max} or 
 \emph{cluster} Voronoi diagram~\cite{egs-ueplf:aa-89,p-cacmmdVLSIc-01,p-hvdpcp-04},
 similarly partitions~$\Plane$ 
 into regions of \emph{nearest} colors with respect to the maximal distances~$\{\bar{d}_i\}_{i \in K}$.
 We thus have $\FCVD(S) = \CVD_{m-1}(S)$ and $\HVD(S) =
 \mCVD_{m-1}(S)$.
 These two diagrams are often refined so that the region of each color~$i\in K$ 
 is subdivided into subregions of
 sites in~$S_i$ that
 determine the distance-to-color~$d_i$ or~$\bar{d}_i$.
 We denote these refined versions by~$\FCVD^*(S)$ and~$\HVD^*(S)$, respectively.  
 See \figurename~\ref{fig:cvd_special}(c) and~(d).
\end{itemize}

In this paper,
we initiate the study of higher-order color Voronoi diagrams, $\CVD_k(S)$ and $\mCVD_k(S)$,
with general distance functions~$\delta_s$ for~$s\in S$.
Our main results are as follows:
\begin{enumerate}[(1)] 
 \item
 We prove an exact upper bound~$4k(n-k)-2n$ on
 the total number of vertices of both order-$k$ diagrams $\CVD_k(S)$ and $\mCVD_k(S)$
 for each $1\leq k \leq m-1$,
 when the sites $s\in S$ are points 
 and $\delta_s(x)$ is given as
 the \emph{convex distance} from~$x\in \Plane$ to~$s$ 
 based on any convex and compact subset of~$\Plane$,
 which includes the $L_p$~distance for any $1 \leq p \leq \infty$.
 This result is in fact a corollary of our more general result:
 we derive an exact equation 
 under certain conditions on functions~$\delta_s$,
 which only requires relations on the numbers
 of vertices and unbounded edges in~$\VD(S')$ and~$\FVD(S')$ for~$S'\subseteq S$.
 The \changed{bound~$4k(n-k)-2n$} 
 can be realized, for example, when $m=n$,
 so \changed{it} is tight and best possible.
 \item
 Under the $L_1$ or the $L_\infty$ metric,
 we prove that $\CVD_k(S)$ and $\mCVD_k(S)$ consist of at most
 $\min\{4k(n-k)-2n, 4(n-k)^2\}$ and $\min\{4k(n-k)-2n, 2k^2\}$ vertices, respectively.
 Similar bounds are derived for
 \changed{any convex distance function based on a convex polygon.}
\item
  We present an iterative algorithm that computes  
  color Voronoi diagrams of order~$1$ 
  to~$k$ in $O(k^2n + n \log n)$ expected or 
  $O(k^2n\log n)$ worst-case time,
  provided that $S$ and $\{\delta_s\}_{s\in S}$ satisfy 
  the requirements of \changed{abstract Voronoi
    diagrams~\cite{k-cavd-89} and}
  an additional condition (see Section~\ref{sec:alg}),
  which includes colored points $S$
  under any smooth convex distance function.
For points~$S$ in the Euclidean plane,
it can be reduced to $O(k^2n + n \log n)$ worst-case time.
\end{enumerate}

Our combinatorial results generalize
previously known bounds for the ordinary higher-order Voronoi diagrams~$\VD_k(S)$
of uncolored sites~$S$,
which is a special case of $m=n$ in our setting.
The asymptotically tight bound~$O(k(n-k))$ on the complexity of~$\VD_k(S)$
has been proved not only
for points~\cite{l-knnvdp-82,lpl-knnvdr-15} under the $L_p$~metric
but also for line segments~\cite{pz-hovdls-16}
and even in the abstract setting~\cite{bcklpz-choavd-15}.
Under the $L_1$ (or $L_\infty$) metric,
the complexity of $\VD_k(S)$ is known to be $O(\min\{k(n-k),(n-k)^2\})$
for a set~$S$ of points or line segments~\cite{lpl-knnvdr-15,pz-hovdls-16}.

In particular, 
the same exact number $4k(n-k)-2n$ can be derived 
for the ordinary order-$k$ diagram~$\VD_k(S)$ and order-$(n-k)$ diagram~$\VD_{n-k}(S)$
under the Euclidean metric from previous results~\cite{e-acg-87,l-knnvdp-82,cs-arscgII-89}.
In his book~\cite[Chapter 13]{e-acg-87},
Edelsbrunner showed that the number of vertices of~$\VD_k(S)$
is \emph{exactly} $(4k-2)n - 2k^2 - e_k - 2 \sum_{i=1}^{k-1} e_{i}$
if $S$ is in general position,
where $e_k$ denotes the number of
\changed{unbounded edges in~$\VD_k(S)$ (or, equivalently, \emph{$k$-sets} in~$S$).}
One can easily verify that the total number of vertices in~$\VD_k(S)$ and~$\VD_{n-k}(S)$
is exactly $4k(n-k) - 2n$
by using the identities: $e_k = e_{n-k}$ and $\sum_{k=1}^{n-1} e_k = n(n-1)$.
This can also be verified from the inductive approach of Lee~\cite{l-knnvdp-82}.
As a recent result related to ours, 
Biswas et al.~\cite{bmes-ccokvtr3mt-21} derived exact relations
for the complexity of 3D Euclidean higher-order Voronoi diagrams with Morse theory,
extending the inductive argument by Lee.

Another remarkable special case of our results is when $k = m-1$,
which yields the $O(m(n-m+1))$ bound 
for the farthest color Voronoi diagram~$\FCVD(S) = \CVD_{m-1}(S)$ 
and the Hausdorff Voronoi diagram~$\HVD(S)=\mCVD_{m-1}(S)$.
The worst-case complexity of~$\FCVD(S)$ and~$\HVD(S)$ is known to be
$\Theta(mn)$ or $\Theta(n^2)$, if $S$ is
a set of points or line segments under the Euclidean 
or $L_1$ metric~\cite{egs-ueplf:aa-89,hks-uevsia-93,ahiklmps-fcvdrp-06,p-hvdpcp-04,b-tbiafcvdls-14}.
While the upper bounds $O(mn)$ and $O(n^2)$ are shown to be tight 
by matching lower bound constructions~\cite{egs-ueplf:aa-89,hks-uevsia-93,ahiklmps-fcvdrp-06},
they become significantly
loose when $n$ is close to~$m$;
if $n=m$, we have
$\FCVD(S) = \FVD(S)$ and $\HVD(S)=\VD(S)$, where both have
linear complexity.
Hence, our new results not only prove \changed{tight}
upper bounds simultaneously
on both diagrams, 
but also formally reveal a smooth extension
upon the previous knowledge about the ordinary Voronoi diagrams for
any~\changed{$m\leq n$.}
Prior to our results, it was known that
the complexity of~$\FCVD(S)$ and $\HVD(S)$ can range from $\Theta(n)$ to $\Theta(m(n-m+1))$
\changed{expressed} in terms of geometric parameters, 
called \emph{straddles} for~$\FCVD(S)$~\cite{MPSW24} 
and \emph{crossings} for~$\HVD(S)$~\cite{p-hvdpcp-04}.
Conditions under which the diagrams have linear complexity
have also been
\changed{discussed~\cite{p-hvdpcp-04,MPSW24,b-lsfcvd-12}.}

Our combinatorial results are based on a color-augmented extension of
the Clarkson--Shor framework~\cite{cs-arscgII-89},
so-called the \emph{colorful Clarkson--Shor framework},
which \changed{has}
its own interest with various applications.
Our notions of colored objects and configurations are naturally inherited
from any
\changed{set} system that fits in the original Clarkson--Shor framework,
and yield a systematic scheme to deal with \changed{objects that are
  collections  of primitive elements.}
Our new framework provides a unifying approach to derive the
complexity of higher-order color Voronoi diagrams
under general distances~$\delta_s$, 
\changed{including the well-studied diagrams $\FCVD(S)$, $\HVD(S)$, and the ordinary higher-order
diagram~$\VD_k(S)$.}
\changed{
  In deriving our results, we make use of a close relation
  between the color diagrams~$\CVD_k(S)$ and $\mCVD_k(S)$
  and \emph{colored} $k$-facets of $S$.}
  An analogous relation for the Euclidean ordinary case (without
  colors) has been shown 
 by Clarkson and Shor~\cite{cs-arscgII-89} and 
 Edelsbrunner~\cite{e-acg-87}. 
\changed{We also} derive lower and upper bounds on the number of colored $k$-facets
in~$\Real^2$.

We demonstrate more
applications of the colorful Clarkson--Shor framework
\changed{that result
in several new bounds on levels of arrangements of 
piecewise linear or algebraic curves and surfaces.}
More specifically,
let $\Gamma$ be a collection of $m$~piecewise algebraic surfaces in~$\Real^d$
and $n$ be their total complexity, counting all algebraic pieces including their boundary elements.
Let $\arr=\arr(\Gamma)$ be their arrangement.
\begin{enumerate}[(1)]
 \setcounter{enumi}{3}
 \item
  If each $\gamma \in \Gamma$ is a convex and monotone polyhedral surface,
  the complexity of the $(\leq k)$-level in~$\arr$
  is shown to be $O(m^{\lfloor d/2\rfloor - 1} k^{\lceil d/2 \rceil} n^{\lfloor d/2 \rfloor})$
  in general, or $k^{\lceil d/2 \rceil} n^{\lfloor d/2 \rfloor}$
  if the number of linear pieces of any two in $\Gamma$ differ at most a constant.
  This extends one of the first results by Clarkson and Show~\cite{cs-arscgII-89} 
  for the arrangement of $n$~hyperplanes,
  and so is tight for~$d \geq 3$.
  Also, note that colored $j$-facets correspond to vertices of~$\arr$
  by the standard point-to-hyperplane duality~\cite{e-acg-87},
  so the same asymptotic upper bounds apply to
  the number of colored $(\leq k)$-facets in a set of $n$ colored points in~$\Real^d$.
 \item If each $\gamma\in \Gamma$ is piecewise linear and monotone,
  then the complexity of the $(\leq k)$-level in~$\arr$ is
  $O(km^{d-2}n^{d-1}\alpha(n/k))$ in general and $O(kn^{d-1}\alpha(n/k))$
  if the number of linear pieces of any two in $\Gamma$ differ at most a constant.
  These bounds are based on the known bound $O(n^{d-1}\alpha(n))$ 
  on the lower envelope of piecewise linear functions~\cite{sa-dsstga-95,e-ueplf:tbnf-89,t-ntasapls-96}.
  In particular, for $d=2$, the bound is reduced to $O(kn\alpha(m/k))$
  by Har-Peled~\cite{h-mcl-99}.
 \item
  Analogously, we obtain upper bounds
  $O(m^{d-2}k^{1-\epsilon}n^{d-1+\epsilon})$ and $O(k^{1-\epsilon}n^{d-1+\epsilon})$
  on the complexity of the $(\leq k)$-level in~$\arr$
  based on the $O(n^{d-1+\epsilon})$ bound on the lower envelope
  of algebraic surface patches 
  under reasonable conditions~\cite{sa-dsstga-95,hs-nblwtd-94,s-atublehd-94}.
  In particular, for~$d=2$, more refined upper bounds 
  of the form $O(kn\beta(n/k))$ or $O(kn\beta(m/k))$
  are obtained,
  where $\beta(\cdot)$ denotes an extremely slowly growing function,
  based on the previous results on the arrangement of Jordan curves~\cite{sa-dsstga-95,h-mcl-99}.
 \item
  Given $m$ convex polygons with a total of $n$ sides in~$\Plane$,
  the \emph{depth} of each point in~$\Plane$ is 
  the number of polygons whose interior contains it.
  Based on results by Aronov and Sharir~\cite{as-cecpp-97},
  we show that the number of vertices of~$\arr$ whose depth is at most~$k$
  is $O((k+1)n\alpha(m/(k+1)) + m^2)$ in general and $O((k+1)n\alpha(m/(k+1)))$
  if the common exterior of any subset of the $m$ polygons is connected.
  Similarly, given $m$ convex polyhedra with a total of $n$~faces in~$\Real^3$,
  we obtain bounds $O((k+1)mn\alpha(m/(k+1)) + m^3)$ and
  $O((k+1)^{1-\epsilon}m^{1+\epsilon}n)$, respectively,
  based on the results on the common exterior of convex polyhedra in~$\Real^3$
  by Aronov et al.~\cite{ast-ucptd-97} and Ezra and Sharir~\cite{es-scacp-07}.
\end{enumerate}

\changed{
Our algorithm follows the principle of iteratively computing all
order-$i$ Voronoi diagrams for $1\leq i\leq k$,
and compares to the $O(k^2n\log n)$-time counterpart of  Lee~\cite{l-knnvdp-82}
for computing $\VD_1(S),\ldots, \VD_k(S)$ for points in the Euclidean metric.}
After a series of improvements~\cite{agss-ltacvdcp-89,as-soriachovd-92,m-lavd-91,adms-clahovd-98,c-rshrrcltd-00,r-rrrsklc-99,ct-oda2d3dsc-16},
the first optimal $O(n\log n + kn)$-time algorithm that computes~$\VD_k(S)$
for points~$S$ in the Euclidean plane
was eventually presented in 2024 by Chan et al.~\cite{ccz-oahovdp:usn-24}.
Efficient algorithms of different approaches are also known
for computing $\VD_k(S)$ of line segments~$S$~\cite{pz-hovdls-16}
or under the model of abstract Voronoi diagrams~\cite{bkl-erahoavd-19,blpz-rdcahoavd-16};
and for computing~$\FCVD^*(S)$~\cite{hks-uevsia-93,ahiklmps-fcvdrp-06,b-tbiafcvdls-14,MPSW24}
\changed{and $\HVD^*(S)$~\cite{egs-ueplf:aa-89,p-hvdpcp-04,AP19}.}

This paper is organized as follows.
In Section~\ref{sec:cvda}, we introduce some preliminary concepts and
make basic observations for~$\CVD_k(S)$ and $\mCVD_k(S)$ in terms of levels of an arrangement of surfaces in~$\Real^3$.
We introduce the colorful Clarkson--Shor framework in Section~\ref{sec:framework}
and prove the complexity bound on the Euclidean higher-order color Voronoi diagrams.
The complexity of higher-order color Voronoi diagrams under general distance functions 
is discussed in Section~\ref{sec:kcvd_general},
and Section~\ref{sec:alg} is devoted to 
our algorithm
to compute higher-order color Voronoi diagrams iteratively.
More applications of the colorful Clarkson--Shor framework
are given in Section~\ref{sec:appl_CS}.
We conclude the paper with some remarks in Section~\ref{sec:conclusion}.

\section{Color Voronoi diagrams and arrangements} \label{sec:cvda}

\changed{Let} $S$ be a set of $n$~\emph{sites}, which can be any abstract objects,
and $\delta_s \colon \Plane \to \Real$ for~$s\in S$
be a given continuous function.
\changed{We assume that the functions $\delta_s$ are}
in \emph{general position},
similarly \changed{to~\cite{kmrss-dpvdgdftaa-20}.}
More precisely,
let $\gamma_s$ be the graph of $\delta_s$, that is,
the $xy$-monotone surface $\{(p, \delta_s(p)) \mid p\in \Plane\}$ in~$\Real^3$,
and $\Gamma := \{\gamma_s \mid s\in S\}$.
\changed{By the general position of} functions~$\delta_s$, we mean \changed{the following:}
no more than three surfaces in~$\Gamma$ meet at a common point,
no more than two surfaces in~$\Gamma$ meet at a one-dimensional curve,
no two surfaces in~$\Gamma$ are tangent to each other, and
none of the surfaces in~$\Gamma$ is tangent to the intersection curve of two others in~$\Gamma$.

We denote the minimization diagram of $\Gamma$ by $\VD(S)$,
the \emph{nearest-site Voronoi diagram} of $S$, and
the maximization diagram of $\Gamma$ by $\FVD(S)$,
the \emph{farthest-site Voronoi diagram} of $S$.
It is well known that the ordinary (uncolored) higher-order Voronoi diagrams~$\VD_k(S)$ 
of~$S$ are determined by
levels of the arrangement~$\arr(\Gamma)$ of~$\Gamma$
\changed{as established in}
earlier work~\cite{es-vda-86,a-pd:paa-87}.
In the following, we discuss an analogous relation between order-$k$ color Voronoi
diagrams and levels of certain surfaces in~$\Real^3$.

Let us assume that the sites in~$S$ are colored with $m$~colors
in~$K = \{1,\ldots, m\}$.
Let $\kappa \colon S \to K$ denote \changed{a}
color assignment such that
the color of~$s\in S$ is $\kappa(s) \in K$.
Let $S_i := \{s\in S \mid \kappa(s) = i\}$
and $\Gamma_i := \{\gamma_s \mid s \in S_i\}$  for $i\in K$.
Define $E_i$ and $\overline{E}_i$ 
to be the lower and upper envelopes of surfaces in~$\Gamma_i$, respectively,
and consider the two arrangements $\arr_\Gamma = \arr(\{E_1,\ldots, E_m\})$ and
$\overline{\arr}_\Gamma = \arr(\{\overline{E}_1,\ldots, \overline{E}_m\})$.
Note that $E_i$ is the graph of the minimal distance function~$d_i$ of color~$i\in K$,
and $\overline{E}_i$ is the graph of the maximal distance~$\bar{d}_i$.
We then consider the levels of the arrangements~$\arr_\Gamma$ and $\overline{\arr}_\Gamma$.
For $1 \leq k \leq m$,
let $L_k$ be the $k$-level of $\arr_\Gamma$ \emph{from below}
and $\overline{L}_k$ be the $k$-level of $\overline{\arr}_\Gamma$ \emph{from above}.
So, $L_1$ is the lower envelope of
$E_1, \ldots, E_m$,
and $\overline{L}_1$ is the upper envelope of
$\overline{E}_1, \ldots, \overline{E}_m$.
Thus, projecting $L_1$ and $\overline{L}_1$ down onto~$\Plane$ yields
$\VD(S)$ and $\FVD(S)$, respectively.
On the other hand, $L_m$ is the upper envelope of the lower envelopes $E_1, \ldots, E_m$, and
$\overline{L}_m$  is the lower envelope of the upper envelopes
$\overline{E}_1, \ldots, \overline{E}_m$.
Projecting $L_m$ and $\overline{L}_m$ down onto~$\Plane$ yields
the refined \changed{diagrams
$\FCVD^*(S)$ and $\HVD^*(S)$,}
\changed{respectively~\cite{hks-uevsia-93,egs-ueplf:aa-89}.}

\begin{figure}[p]
\begin{center}
\includegraphics[width=.95\textwidth]{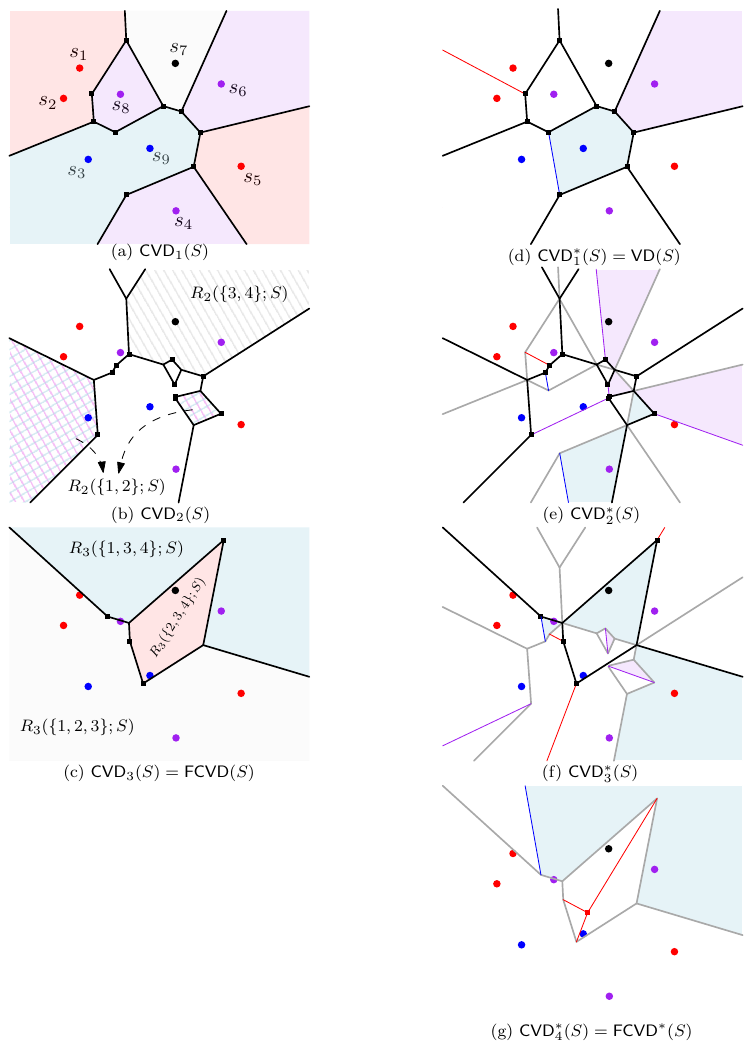}
\end{center}
\caption{The minimal color Voronoi diagrams~$\CVD_k(S)$ and 
 the refined diagrams~$\CVD^*_k(S)$.
 }
\label{fig:cvd}
\end{figure}

\begin{figure}[p]
\begin{center}
\includegraphics[width=.95\textwidth]{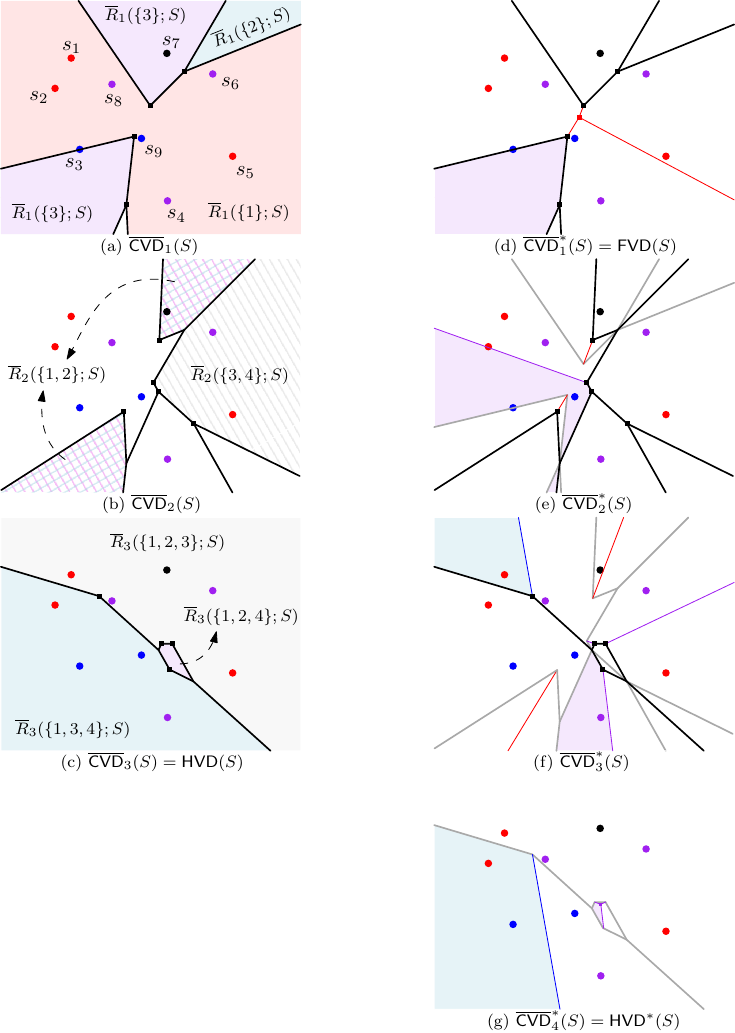}
\end{center}
\caption{The maximal color diagrams~$\mCVD_k(S)$ and 
  the refined diagrams~$\mCVD^*_k(S)$.
 }
\label{fig:mcvd}
\end{figure}

From the viewpoint of~\cite{es-vda-86,a-pd:paa-87},
we observe the following.

\begin{itemize}
\item The order-$k$ color Voronoi diagrams, $\CVD_k(S)$ and $\mCVD_k(S)$, for~$1\leq k\leq m-1$
 are the projections of $L_k\cap L_{k+1}$ and of $\overline{L}_k\cap \overline{L}_{k+1}$,
 respectively, onto~$\Plane$.
\item For each $1\leq k \leq m$,
 let $\CVD^*_k(S)$ denote the planar map obtained
 by projecting $L_k$ down onto~$\Plane$;
 analogously, let $\mCVD^*_k(S)$ be the \changed{map}
 obtained by projecting $\overline{L}_k$ onto~$\Plane$.
 By definition, 
 $\CVD^*_k(S)$ and $\mCVD^*_k(S)$ refine $\CVD_k(S)$ and $\mCVD_k(S)$,
 respectively.
 Each face~$f$ of~$\CVD^*_k(S)$ (or of~$\mCVD^*_k(S)$)
 is \emph{associated with} a site~$s\in S_i$ such that, for any point~$p\in f$, 
 $d_i(p) = \delta_s(p)$ 
 and $i\in K$ is the \emph{$k$-th nearest color} from~$p$ with respect to~$\{d_i\}_{i\in K}$
 (resp. $\bar{d}_i(p) = \delta_s(p)$ and $i\in K$ is the \emph{$k$-th farthest color} from~$p$
 with respect to~$\{\bar{d}_i\}_{i\in K}$).
 This way, the refined diagrams~$\CVD^*_k(S)$ and $\mCVD^*_k(S)$ partition the plane~$\Plane$ by 
 the $k$-th nearest and \changed{$k$-th} farthest colors, respectively.
\end{itemize}

From the construction,
it is clear that $\CVD^*_1(S) = \VD(S)$ and $\mCVD^*_1(S) = \FVD(S)$.
Note also that $\CVD^*_m(S)=\FCVD^*(S)$ and $\mCVD^*_m(S)=\HVD^*(S)$, 
whereas $\FCVD(S) = \CVD_{m-1}(S)$ and $\HVD(S) =\mCVD_{m-1}(S)$.

\changed{\figurename s~\ref{fig:cvd} and~\ref{fig:mcvd} illustrate an
  example under the Euclidean metric, where $S$ consists of
  $n=9$~points and $m=4$~colors:
  }
$S_1 = \{s_1,s_2,s_5\}$, $S_2 =\{s_3,s_9\}$, $S_3 =\{s_4,s_6,s_8\}$, and $S_4 =\{s_7\}$,
in red, blue, purple, and black, respectively;
selected regions of~$\CVD_k(S)$ and~$\mCVD_k(S)$ are labeled in (a)--(c); 
faces associated with~$s_6\in S_3$ and those with~$s_9\in S_2$ in~$\CVD^*_k(S)$ and~$\mCVD^*_k(S)$
are shaded in purple and blue, respectively, in (d)--(g).

Now, consider the vertices and edges of the arrangements~$\arr_\Gamma$ and $\overline{\arr}_\Gamma$.
By the general position assumption,
each vertex~$v$ of~$\arr_\Gamma$ or of~$\overline{\arr}_\Gamma$
is determined by exactly three sites $s,s',s''\in S$
in such a way that $v$ is a common intersection of
three surfaces~$\gamma_{s}, \gamma_{s'}, \gamma_{s''} \in \Gamma$.
Such a vertex~$v$ is called \emph{$c$-chromatic} for~$c\in \{1,2,3\}$
if $|\{\kappa(s), \kappa(s'), \kappa(s'')\}| = c$.
(Note that any $1$-chromatic vertex is a vertex of some single envelope~$E_i$ or $\overline{E}_i$.)
Similarly, each edge of~$\arr_\Gamma$ and of~$\overline{\arr}_\Gamma$ is
determined by exactly two sites in~$S$, and is
either $1$- or $2$-chromatic according to the number of
involved colors.
We identify each vertex or edge of $\CVD^*_k(S)$ or of $\mCVD^*_k(S)$ by
its original lifted copy in~$\arr_\Gamma$ or in $\overline{\arr}_\Gamma$.
Observe that any $c$-chromatic vertex or edge of $\arr_\Gamma$ or of $\overline{\arr}_\Gamma$
appears in $c$~consecutive levels,
as it lies in the intersection of $c$ surfaces from $\{E_i\}_{i\in K}$ 
or from $\{\overline{E}_i\}_{i\in K}$, respectively.
Thus, $c$-chromatic vertices appear in $c$~consecutive orders of the refined diagrams.
We call a vertex or an edge of~$\CVD^*_k(S)$ or of~$\mCVD^*_k(S)$
\emph{new}
if it does not appear in~$\CVD^*_{k-1}(S)$ or in~$\mCVD^*_{k-1}(S)$,
respectively; \changed{and}
\emph{old}, otherwise.
By definition, 
the vertices and edges of~$\CVD^*_1(S)$ and~$\mCVD^*_1(S)$ are all new.
Note that every edge of $\CVD_k(S)$ (or, of $\mCVD_k(S)$) is $2$-chromatic and new,
and appears both in~$\CVD^*_k(S)$ and~$\CVD^*_{k+1}(S)$ 
(in $\mCVD^*_k(S)$ and $\mCVD^*_{k+1}(S)$, resp.), being first new and then old.
See~\figurename s~\ref{fig:cvd} and~\ref{fig:mcvd}, where
new $2$-chromatic edges are in black,
old $2$-chromatic edges in gray,
$1$-chromatic edges in their own color, and
new vertices marked by small squares.
 
We define $v_{c,j}=v_{c,j}(S)$ for~$1\leq c \leq 3$ and~$0\leq j \leq m-1$ 
to be the number of $c$-chromatic vertices in~$\arr_\Gamma$
below which there are exactly $j$~surfaces from~$\{E_i\}_{i\in K}$;
$\bar{v}_{c,j}=\bar{v}_{c,j}(S)$
to be the number of  $c$-chromatic vertices in~$\overline{\arr}_\Gamma$
above which there are exactly $j$~surfaces from~$\{\overline{E}_i\}_{i\in K}$.

\begin{lemma} \label{lem:CVD_v}
 For any $1\leq c\leq 3$ and $1\leq k\leq m$, the following hold: 
 \begin{enumerate}[(i)] 
 \item 
 The number of new $c$-chromatic vertices of~$\CVD^*_{k}(S)$ 
 is exactly~$v_{c,k-1}$, and
 the number of new $c$-chromatic vertices of~$\mCVD^*_k(S)$ is exactly~$\bar{v}_{c,k-1}$.
 \item The number of vertices of~$\CVD^*_k(S)$ is exactly
 $v_{3,k-1} + v_{3,k-2} + v_{3,k-3} + v_{2,k-1} + v_{2,k-2} + v_{1,k-1}$,
 and the number of vertices of~$\mCVD^*_k(S)$ is exactly
 $\bar{v}_{3,k-1}+\bar{v}_{3,k-2}+\bar{v}_{3,k-3} + \bar{v}_{2,k-1}+\bar{v}_{2,k-2} + \bar{v}_{1,k-1}$,
 where $v_{c,j} = \bar{v}_{c,j} = 0$ for $j < 0$.
 \item  The number of vertices of~$\CVD_k(S)$ is exactly
 $v_{3,k-1} + v_{3,k-2} + v_{2,k-1}$, and
 the number of vertices of~$\mCVD_k(S)$ is exactly
 $\bar{v}_{3,k-1} + \bar{v}_{3,k-2} + \bar{v}_{2,k-1}$.
 \end{enumerate}
\end{lemma}

\begin{proof}
Consider any new vertex~$v$ in~$\CVD^*_k(S)$ defined by three distinct sites~$s, s', s''\in S$.
Each of the colors~$\kappa(s), \kappa(s'), \kappa(s'')$
is the $k$-th nearest color at~$v\in \Plane$.
Hence, there are exactly $k-1$ additional colors~$i \in K\setminus \{\kappa(s),\kappa(s'),\kappa(s'')\}$
such that the following strict inequality holds:
\[ d_i(v) < d_{\kappa(s)}(v) = d_{\kappa(s')}(v) = d_{\kappa(s'')}(v). \]
Analogously,
for each new vertex~$v$ of~$\mCVD^*_k(S)$ defined by $s, s', s''\in S$,
there are exactly $k-1$ additional colors~$i \in K\setminus \{\kappa(s),\kappa(s'),\kappa(s'')\}$ 
such that
\[ \bar{d}_i(v) > \bar{d}_{\kappa(s)}(v) = \bar{d}_{\kappa(s')}(v) = \bar{d}_{\kappa(s'')}(v).\]
Hence,
there is a one-to-one correspondence between
 new $c$-chromatic vertices in~$\CVD^*_{k}(S)$ and
 $c$-chromatic vertices of~$\arr_\Gamma$, below which there are exactly $k-1$~surfaces
 from~$\{E_i\}_{i\in K}$.
 Analogously, there is another correspondence between
 new $c$-chromatic vertices in~$\mCVD^*_{k}(S)$ and
 $c$-chromatic vertices of~$\overline{\arr}_\Gamma$, above which there are exactly $k-1$~surfaces
 from~$\{\overline{E}_i\}_{i\in K}$.
This proves claim~(i).

Since $c$-chromatic vertex or edge is contained in $c$ consecutive levels in~$\arr_\Gamma$
or in~$\overline{\arr}_\Gamma$, we know that
$\CVD^*_k(S)$ (or $\mCVD^*_k(S)$) consists of 
new $3$-chromatic vertices in order-$k$, order-$(k-1)$, or order-$(k-2)$;
new $2$-chromatic vertices in order-$k$ or order-$(k-1)$; and
new $1$-chromatic vertices in order-$k$.
This implies claim~(ii).

From the above discussions,
we know that the vertices of~$\CVD_k(S)$ (resp. of~$\mCVD_k(S)$) are those that appear commonly 
in~$\CVD^*_k(S)$ and~$\CVD^*_{k+1}(S)$ (resp. in~$\mCVD^*_k(S)$ and~$\mCVD^*_{k+1}(S)$).
Therefore,
$\CVD_k(S)$ (resp. $\mCVD_k(S)$) consists of
new $3$-chromatic vertices in order-$k$ or order-$(k-1)$ and
new $2$-chromatic vertices in order-$k$,
so claim~(iii) follows.
Note that $v_{3,m-1} = v_{3,m-2} = v_{2,m-1} = 0$ by definition,
so claim~(iii) concludes that $\CVD_m(S)$ and $\mCVD_m(S)$ have zero vertex,
which is certainly true
since both $\CVD_m(S)$ and $\mCVD_m(S)$ consist of a single face
$\VR_m(K; S) = \mVR_m(K; S) = \Plane$.
\end{proof}

\section{The colorful Clarkson--Shor framework} \label{sec:framework}

The Clarkson--Shor technique~\cite{cs-arscgII-89} is based on
a general framework dealing with so-called configurations or ranges 
defined by a set of objects.
Specifically, the following three ingredients are supposed to be given
with a constant integer parameter~$d\geq 1$, see also Sharir~\cite{s-cstre-03}:
\begin{itemize} 
 \item A set~$S$ of $n$~\emph{objects}.
 \item A set~$\conf(S)$ of \emph{configurations},
  each of which is defined by exactly $d$~objects in~$S$.
 \item A \emph{conflict relation}~$\chi \subseteq S \times \conf(S)$
  between objects~$s\in S$ and configurations~$f \in \conf(S)$
  with the requirement that none of the $d$~objects defining~$f$ are in conflict with~$f$.
\end{itemize}
In the original framework, the number of objects that define a configuration
does not have to be exactly~$d$, but at most~$d$.
This restriction can be achieved by adding dummy objects to~$S$.

Let us call such a \changed{triplet}
$(S, \conf(S), \chi)$ a \emph{CS-structure}.
Given any CS-structure~$(S, \conf(S), \chi)$ with parameter~$d$,
we now impose a \emph{color} assignment
$\kappa \colon S \to K$ to the objects in~$S$,
where $K = \{1,2,\ldots, m\}$ denotes the set of $m$ colors with $m\leq n$.
For each color $i\in K$, let $S_i := \{s \in S \mid \kappa(s) = i\}$.
For~$f\in \conf(S)$ and
set $D_f \subseteq S$ of $d$~objects defining~$f$,
$\kappa(D_f) = \{ \kappa(s) \mid s\in D_f\}$ is called a
set of colors defining~$f$.
We build a \emph{color-to-configuration conflict relation} $\chi_\kappa \subseteq K \times \conf(S)$
such that a color $i\in K$ is in conflict with a configuration $f\in \conf(S)$
if an object $s\in S_i$ is in conflict with $f$, 
that is, $(i,f) \in \chi_\kappa$ if and only if $(s, f) \in \chi$ for some $s\in S_i$.
We are then interested in those configurations~$f \in \conf(S)$
such that 
none of its defining colors in $\kappa(D_f)$ are in conflict \changed{with~$f$.}
Let $\conf(S,\kappa) \subseteq \conf(S)$ be the set of these configurations,
called \emph{colored configurations with respect to $\kappa$}.
We call $f\in \conf(S,\kappa)$ \emph{$c$-chromatic} if $|\kappa(D_f)|
= c$ for $1\leq c \leq d$,
\changed{and let the \emph{weight} of~$f$ be} 
the number of colors in~$K$ that are in conflict with~$f$.
Let $\conf_{c,j}(S,\kappa)\subseteq \conf(S,\kappa)$ be
the set of $c$-chromatic weight-$j$ colored configurations \changed{in~$\conf(S,\kappa)$.}

\changed{Considering the colors in~$K$ as new \emph{objects}, each of which is a collection of objects in~$S$,}  
observe that this color-augmented structure $(K, \bigcup_{j}\conf_{c,j}(S,\kappa), \chi_\kappa)$
for each $1\leq c\leq d$ is again a CS-structure with parameter~$c$.
Therefore, the main lemma of Clarkson and Shor~\cite[Lemma~2.1]{cs-arscgII-89}
automatically implies the following.
\begin{lemma}\label{lem:CS_general_lb}
 With the above notations, 
 let $1\leq c\leq d$, $r \geq 0$, and $a \geq 0$ be integers,
 and $R \subseteq K$ be a random subset of $r$ colors.
 Then,
  \[ \binom{m}{r} \E[ |\conf_{c,a}(S_R, \kappa_R)| ] \geq
      \sum_{j=0}^{m-c} |\conf_{c,j}(S, \kappa)| \binom{j}{a}\binom{m-c-j}{r-c-a},\]
 where $S_R = \bigcup_{i\in R}S_i$ and
 $\kappa_R \colon S_R \to R$ denotes the restriction of $\kappa$ to $S_R$.
 The equality holds if each configuration in~$\conf(S, \kappa)$ is defined 
 by a unique set of $d$~objects in~$S$.
\end{lemma}

In probabilistic arguments dealing with CS-structures,
it is usually necessary to have an upper bound on
the number of weight-$0$ configurations.
For any subset $S'\subseteq S$,
let $\conf_0(S') \subseteq \conf(S')$ be the set of (uncolored)
configurations~$f$ \changed{of weight 0},
that is, $(s,f)\notin \chi$ for all~$s\in S'$.
Let $T_0(n')$ be 
any nondecreasing function with $T_0(0) = 0$ that upper bounds
the number $|\conf_0(S')|$ of \changed{these}
configurations for any set $S'$ of $n'$
uncolored objects.
The following is obvious by definition.
\begin{lemma} \label{lem:CS_general_ub}
 Let $S' \subseteq S$ be any subset and $\kappa'$ be any color assignment for $S'$.
 Then,
 \[ \sum_{c=1}^d |\conf_{c,0}(S', \kappa')| = |\conf_0(S')| \leq T_0(|S'|).\]
\end{lemma}

In many applications, 
such an upper bound function~$T_0$ is either known from previous work
or relatively easy to obtain.
Once we have~$T_0$, 
we can derive general upper bounds
on the number of corresponding colored configurations
of weight at most~$k$ in such a procedural way as done for uncolored
cases~\cite{sa-dsstga-95,cs-arscgII-89,m-ldg-02}.

\added{
The following observation will be useful for this purpose.
\begin{lemma} \label{lem:sums}
 Let $A$ be a finite set of real numbers and $r$ be an integer.
 Then, it holds that
 \[ \sum_{R \subseteq A, |R| = r} \sum_{a\in R} a = \binom{|A|-1}{r-1} \sum_{a \in A} a. \]
\end{lemma}
\begin{proof}
Note that the sum runs over all $r$-subsets of~$A$.
For any fixed $a\in A$,
we observe that out of $\binom{|A|}{r}$ $r$-subsets of~$A$,
exactly $\binom{|A|-1}{r-1}$ subsets contain $a$.
Hence, the lemma follows.
\end{proof}
}

\begin{theorem} \label{thm:CS_general}
 With the above notations,
 suppose $T_0$ is a convex function. 
 For each $1\leq c\leq d$,
 the total number of $(\leq c)$-chromatic colored configurations is bounded by
 \[ \sum_{b=1}^c\sum_{j=0}^{m-b} \binom{m-b-j}{c-b}|\conf_{b,j}(S,\kappa)| \leq
   \binom{m}{c}\cdot\frac{1}{m} \sum_{i\in K} T_0(c|S_i|)
   = O(m^{c-1}\cdot T_0(cn)).\]
 Also, for each $2\leq c\leq d$ and $0\leq k \leq \lfloor \frac{m}{c} \rfloor - 1$,
it holds that
 \[ \sum_{j=0}^k |\conf_{c,j}(S,\kappa)| =
  O\left(\frac{(k+1)^c}{m}\cdot\sum_{i\in K} T_0\left(\frac{m |S_i|}{k+1}\right)\right)
  = O\left(\frac{(k+1)^c}{m}\cdot T_0\left(\frac{mn}{k+1}\right)\right).
 \]
\end{theorem}
\begin{proof}
Let $1\leq r \leq m$ be an integer parameter
and $R\subseteq K$ be a random subset of $r$ colors.
We start by showing an upper bound on $\E[|\conf_{c,0}(S_R, \kappa_R)|]$.
By Lemma~\ref{lem:CS_general_ub}, observe that
\[ \binom{m}{r} \sum_{c=1}^d \E[ |\conf_{c,0}(S_R, \kappa_R)| ] \leq \binom{m}{r} \E[T_0(|S_R|)] =
     \sum_{R'\subseteq K, |R'|=r} T_0(|S_{R'}|).\]
Since $T_0(n)$ is a convex function,
we have Jensen's inequality: 
\[ T_0\left(\sum_{a\in A} a\right) \leq \frac{1}{|A|}\cdot \sum_{a\in A} T_0(|A|\cdot a)\]
for any finite set~$A$ of positive real numbers.
We thus have
\begin{align*}
 \binom{m}{r} \sum_{c=1}^d \E[ |\conf_{c,0}(S_R, \kappa_R)| ] &\leq
  \sum_{R'\subseteq K, |R'|=r} T_0(|S_{R'}|) 
  = \sum_{R'\subseteq K, |R'|=r} T_0\left( \sum_{i\in R'} |S_i|\right) \\
  &\leq \sum_{R'\subseteq K, |R'|=r} \sum_{i\in R'} T_0(r\cdot |S_i|)/r\\
  & = \binom{m-1}{r-1} \frac{1}{r} \sum_{i\in K} T_0(r\cdot |S_i|)  \\
  & = \binom{m}{r} \frac{1}{m} \sum_{i\in K} T_0(r\cdot |S_i|)
\end{align*}
by Lemma~\ref{lem:sums}.
Hence, we have
\[
 \sum_{c=1}^d \E[|\conf_{c,0}(S_R, \kappa_R)|] \leq \frac{1}{m} \sum_{i\in K} T_0(r\cdot |S_i|).\]

On the other hand, Lemma~\ref{lem:CS_general_lb} (for~$a=0$) implies
\[ \E[|\conf_{c,0}(S_R, \kappa_R)|] \geq \sum_{j=0}^{m-c} |\conf_{c,j}(S, \kappa)|
  \cdot \binom{m-c-j}{r-c} \Big/ \binom{m}{r},\]
for $1\leq c \leq d$.

To obtain the first bound,
we fix $c$ with $1\leq c\leq d$ and set $r=c$.
We then have
\[ \sum_{b=1}^c \sum_{j=0}^{m-b} |\conf_{b,j}(S, \kappa)| \cdot \binom{m-b-j}{c-b} \Big/ \binom{m}{c}
   \leq \sum_{b=1}^c \E[|\conf_{b,0}(S_R, \kappa_R)|]  \leq 
   \frac{1}{m} \sum_{i\in K} T_0(r\cdot |S_i|)\]
by plugging the above lower and upper bounds. 

For the second bound,
let $2\leq c\leq d$ and $k \leq \lfloor \frac{m}{c} \rfloor - 1$ be fixed,
and set $r = \lfloor \frac{m}{k+1} \rfloor$.
The factor $\binom{m-c-j}{r-c} \Big/ \binom{m}{r}$ is then lower bounded as follows.
(Almost the same derivation can be found in Matou\v{s}ek~\cite[Lemma 6.3.2]{m-ldg-02}.)
\begin{align*}
  \binom{m-c-j}{r-c}  \Big/ \binom{m}{r}
 & = \frac{r(r-1)\cdots(r-c+1)}{m(m-1)\cdots(m-c+1)} \cdot
  \frac{m-c-j}{m-c}\cdot \frac{m-c-1-j}{m-c-1}\cdots\frac{m-r+1-j}{m-r+1} \\
 & = \frac{r(r-1)\cdots(r-c+1)}{m(m-1)\cdots(m-c+1)} \cdot
   \left(1-\frac{j}{m-c}\right)\left(1-\frac{j}{m-c-1}\right)\cdots\left(1-\frac{j}{m-r+1}\right)\\
 & \geq \frac{r(r-1)\cdots(r-c+1)}{m(m-1)\cdots(m-c+1)} \cdot
   \left(1-\frac{k}{m-r+1}\right)^r. \\
\end{align*}
Since $\frac{k}{m-r+1} \leq \frac{k+1}{m}$ and
$1-x \geq (\frac{c-1}{c})^{cx}$ for $0\leq x \leq \frac{1}{c}$,
we have
 \[ \left(1-\frac{k}{m-r+1}\right)^r \geq \left(1-\frac{k+1}{m}\right)^r
    \geq \left(\frac{c-1}{c}\right)^{cr\frac{k+1}{m}}
    \geq \left(\frac{c-1}{c}\right)^{c},\]
as $\frac{k+1}{m} \leq \frac{1}{c}$ and $r = \lfloor \frac{m}{k+1} \rfloor$.
We thus conclude that
 \[ \E[|\conf_{c,0}(S_R, \kappa_R)|] \geq \sum_{j=0}^{k} |\conf_{c,j}(S, \kappa)| \cdot
 \left(\frac{c-1}{c}\right)^{c} \cdot \frac{r(r-1)\cdots(r-c+1)}{m(m-1)\cdots(m-c+1)}.\]
Combining this with the above upper bound 
 \[ \E[|\conf_{c,0}(S_R, \kappa_R)|] \leq \sum_{b=1}^{d} \E[|\conf_{b,0}(S_R, \kappa_R)|]
  \leq \frac{1}{m}\sum_{i\in K} T_0(r\cdot |S_i|)\]
results in 
 \[ \sum_{j=0}^k |\conf_{c,j}(S, \kappa)| 
  = O\left(\frac{(k+1)^c}{m} \cdot \sum_{i\in K} T_0\left(\frac{m |S_i|}{k+1}\right)\right)\]
since $r = \lfloor \frac{m}{k+1}\rfloor$.
As $T_0$ is a convex, nondecreasing function with~$T_0(0) = 0$ and $\sum_{i\in K} |S_i| = n$,
we have $\sum_{i\in K} T_0\left(\frac{m |S_i|}{k+1}\right) \leq T_0\left(\frac{mn}{k+1}\right)$,
so the second bound holds for each $2\leq c\leq d$.
\end{proof}

Remark that the left-hand side of the first bound can be seen as
\changed{a ``weighted''} count of $(\leq c)$-chromatic colored configurations,
and that the second bound in Theorem~\ref{thm:CS_general} for~$n=m$ implies
Clarkson and Shor's original bound, $O((k+1)^c\cdot T_0(n/(k+1)))$,
for uncolored cases~\cite[Theorem~3.1]{cs-arscgII-89}.
Note that
Theorem~\ref{thm:CS_general} still implies the same Clarkson--Shor bound
if $T_0$ is linear.
Moreover, if the colors are assigned in a favorably uniform way,
we can derive similar bounds as well without assuming the convexity of~$T_0$.

\begin{theorem} \label{thm:CS_general_uniform}
 With the above notations, 
 suppose 
 $|S_i| \leq \rho \cdot \frac{n}{m}$ for every~$i\in K$ for some~$\rho \geq 1$.
 Then, for $1\leq c\leq d$,
 \[ \sum_{b=1}^c\sum_{j=0}^{m-b} \binom{m-b-j}{c-b}|\conf_{b,j}(S,\kappa)|
  \leq \binom{m}{c} T_0\left(c \rho \cdot \frac{n}{m} \right),\]
 and for $2\leq c\leq d$ and $0\leq k\leq \lfloor\frac{m}{c}\rfloor - 1$,
 \[ \sum_{j=0}^k |\conf_{c,j}(S,\kappa)| 
  = O\left((k+1)^c\cdot T_0\left( \rho \cdot \frac{n}{k+1} \right)\right).\]
\end{theorem}

\begin{proof}
By Lemma~\ref{lem:CS_general_ub}, we have
\begin{align*}
 \sum_{c=1}^{d} \E[ |\conf_{c,0}(S_R, \kappa_R)| ]  
  &\leq \E[T_0(|S_R|)] \\
  &\leq T_0\left(r \cdot \rho \cdot \frac{n}{m} \right)
\end{align*}
since $T_0$ is nondecreasing and $|S_{R'}| \leq r\cdot \rho \frac{n}{m}$
by the assumption.
The theorem follows by almost the same arguments as in the proof of Theorem~\ref{thm:CS_general},
exploiting the lower bound shown in Lemma~\ref{lem:CS_general_lb}.
\end{proof}

\subsection{Colored \texorpdfstring{$j$-facets}{j-facets}}
Let $S$ be a set of $n$ points in~$\Real^d$.
\changed{
A \emph{$j$-facet} in~$S$ is an oriented $(d-1)$-simplex~$\sigma$ with its vertices chosen from~$S$
such that the open half-space on its positive side contains exactly $j$~points of~$S$.}
Among the first applications of the original Clarkson--Shor framework
was the tight upper bound $O(k^{\lceil d/2\rceil}n^{\lfloor d/2 \rfloor})$
on the number of $(\leq k)$-facets~\cite{cs-arscgII-89},
implying the same asymptotic bound on the $(\leq k)$-level
in the arrangement of hyperplanes via the point-to-hyperplane duality~\cite{e-acg-87}.
Many variants of $j$-facets have been discussed
in the literature;
we refer to a survey article by Wagner~\cite{w-kskf-08}.

\begin{figure}[htb]
\begin{center}
\includegraphics[width=.55\textwidth]{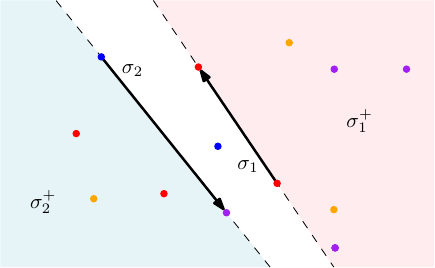}
\end{center}
\caption{Colored $j$-facets in colored points in~$\Plane$:
a $1$-chromatic $2$-facet~$\sigma_1$  
and a $2$-chromatic $2$-facet~$\sigma_2$ are shown, 
which choose half-planes~$\sigma_1^+$ and $\sigma_2^+$ on their right side.
}
\label{fig:j-facets}
\end{figure}

Now, we assume that the points in~$S$ are colored 
by a color assignment~$\kappa$ with $m$ colors~$K$.
For any subset $A \subset \Real^d$,
we shall say that $A$ \emph{intersects} a color $i \in K$,
\changed{if $A$ contains some $s\in S_i$.}
According to our notion of colored configurations,
a \emph{colored $j$-facet}~$\sigma$ in~$S$ with respect to~$\kappa$ 
is an oriented simplex defined by $d$~points $D_\sigma\subseteq S$
such that exactly $j$ colors,
but none of the defining colors in~$\kappa(D_\sigma)$,
are intersected by~$\sigma^+$.
See \figurename~\ref{fig:j-facets}.
Notice that
colored $j$-facets correspond to vertices of the arrangement
of $m$~lower/upper envelopes of hyperplanes dual to~$S_i$ for $i\in K$;
see \changed{Section~\ref{subsec:envelopes_hyperplanes}}
for \changed{a more detailed discussion.}

For $1\leq c\leq d$ and $j \geq 0$,
let $e_{c,j}(S)$ be the number of $c$-chromatic $j$-facets in \changed{$S$}
and $e_j(S) := \sum_{c} e_{c,j}(S)$ be the number of all $j$-facets.
Katoh and Tokuyama~\cite[Proposition 15]{kt-klcs-02} proved that $e_k(S) = O(k^{1/3}n)$ in~$\Real^2$
and $e_k(S) = O(k^{2/3}n^2)$ in~$\Real^3$
based on a generalized Lov\'{a}sz's Lemma.
Theorems~\ref{thm:CS_general} and~\ref{thm:CS_general_uniform} directly 
imply the following bounds.
\begin{corollary} \label{coro:j-facets_general} 
 For a set~$S$ of $n$~colored points in~$\Real^d$ with $m$ colors 
 and any~$0\leq k\leq \lfloor \frac{m}{d} \rfloor -1$,
 the number of $(\leq k)$-facets in~$S$ is
 $\sum_{j=0}^k e_j(S) = O(m^{\lfloor d/2 \rfloor -1} k^{\lceil d/2 \rceil} n^{\lfloor d/2 \rfloor})$.
 If there is a constant $\rho \geq 1$ such that 
 $|S_i| \leq \rho\cdot\frac{n}{m}$ for every~$i\in K$,
 then the bound is improved to
 $\sum_{j=0}^k e_j(S) = O(k^{\lceil d/2 \rceil} n^{\lfloor d/2 \rfloor})$.
\end{corollary}
\begin{proof}
The number of $0$-facets in $n$ (uncolored) points in~$\Real^d$
is exactly the number of facets of the convex hull of the $n$~points,
so we can take $T_0(n) = C_d n^{\lfloor d/2\rfloor}$ for
some constant~$C_d$ depending only on~$d$.
Hence, Theorems~\ref{thm:CS_general} and~\ref{thm:CS_general_uniform}
imply the claimed bounds.
\end{proof}

Note that for large $k$ with $k \geq \lfloor \frac{m}{d} \rfloor$,
the above bound on
the number of $(\leq k)$-facets is asymptotically the same as
the total number of configurations
(see Theorems~\ref{thm:CS_general} and~\ref{thm:CS_general_uniform}).

We continue our discussion for the case of $d=3$.
Lemma~\ref{lem:CS_general_ub} implies that
$e_0(S)$ counts the number of facets of the convex hull of~$S$ in~$\Real^3$.
Using this fact, we observe the following.
\begin{lemma} \label{lem:j-facets_sampling_ub} 
 Let $2\leq r\leq m$
 and $R\subseteq K$ be a random set of $r$ colors.
 It holds that
 \[ \binom{m}{r}\E[e_0(S_R)] \leq 2\binom{m-1}{r-1}n - 4\binom{m}{r}\]
 with equality when the points in $S$ are in convex and general
 position.
\end{lemma}
\begin{proof}
Recall that $e_0(S_R)$ counts the facets of the convex hull of~$S_R$.
By considering all possible $r$-subsets~$R' \subseteq  K$,
we have
\[
 \binom{m}{r}\E[e_0(S_R)] = \sum_{R'\subseteq K, |R'| = r} e_0(S_{R'}).
\]
Let $N_{R'} := |S_{R'}|$.
In~$\Real^3$, we have $e_0(R') \leq 2N_{R'} - 4$
if $N_{R'} \geq 4$,
and the equality holds if the points in $S_{R'}$
are in convex and general position~\cite[Theorem~6.11]{e-acg-87}.
The cases of $N_{R'} < 4$ are handled as follows.
If $N_{R'} = 3$, then we have $e_0(R') = 2$
since the only triangle defined by the three points determines
exactly two $0$-facets;
if $N_{R'} \leq 2$, then we have $e_0(R') = 0$.
Hence, for $r\geq 2$, we have $N_{R'} \geq 2$ for every $r$-subset $R'\subseteq K$ and
it thus holds that $e_0(S_{R'}) \leq 2N_{R'}-4$.
Thus, we have
\begin{align*}
\binom{m}{r}\E[e_0(S_R)] & \leq \sum_{R' \subseteq K, |R'| = r}
    \left(2 N_{R'} - 4\right) \\
   & = 2\sum_{R'} N_{R'} - 4\binom{m}{r} \\
   & = 2\binom{m-1}{r-1}n - 4\binom{m}{r}.
\end{align*}
The last derivation is due to Lemma~\ref{lem:sums}.
Moreover, the equality holds if the points in $S$ are in convex and general position,
as discussed above.
\end{proof}

Now, suppose that $S$ is in convex and general position.
Then, Lemmas~\ref{lem:CS_general_lb} and~\ref{lem:j-facets_sampling_ub}
provide two different ways of exactly counting $\binom{m}{r} \E[e_0(S_R)]$
for $2\leq r \leq m$, 
resulting in:
\begin{theorem} \label{thm:j-facets_convex_3d} 
 Let $S \subset \Real^3$ be a set of $n$ points in convex and general position,
 each of which is colored by one of $m$ colors.
 Then, it holds that for each $0\leq j \leq m-2$
  \[ e_{3,j}(S) + \sum_{i=0}^j e_{2,i}(S) + \sum_{i=0}^j (j-i+1) e_{1,i}(S) = 2(j+1)(n-j-2).\]
\end{theorem}
\begin{proof}
Let $R$ be a random $r$-subset of $K$ with $2 \leq r \leq m$.
Throughout this proof, we write $e_{c,j} = e_{c,j}(S)$ for simplicity.
Since the points in~$S$ are in general position,
Lemma~\ref{lem:CS_general_lb} implies the following equation.
\[ \binom{m}{r} \E[e_0(S_R)] = \binom{m}{r} \sum_{c=1}^3 \E[e_{c,0}(S_R)] =
     \sum_{c=1}^{3} \sum_{j=0}^{m-c} \binom{m-c-j}{r-c} e_{c,j}.\]
Together with Lemma~\ref{lem:j-facets_sampling_ub},
we obtain the following $m-1$ equations:
\[ \sum_{c=1}^{3} \sum_{j=0}^{m-c} \binom{m-c-j}{r-c} e_{c,j}
     = 2\binom{m-1}{r-1} n - 4\binom{m}{r}\]
for $2\leq r \leq m$,
since the points in~$S$ are in convex and general position.

Then, the one for $r=2$ writes
\[ \sum_{i=0}^{m-2} e_{2,i} + \sum_{i=0}^{m-2} (m-1-i) e_{1,i}
   = 2\binom{m-1}{1}n - 4\binom{m}{2} = 2(m-1)(n-m),\]
which is the claimed equation for $j=m-2$, since $e_{3,m-2} = 0$.

Rearranging the other $m-2$ equations for $3\leq r \leq m$, we have
\[ \sum_{j=0}^{m-3} \binom{m-3-j}{r-3} e_{3,j} = 2\binom{m-1}{r-1} n - 4\binom{m}{r} -
      \sum_{j=0}^{m-2} \binom{m-2-j}{r-2} e_{2,j} - \sum_{j=0}^{m-1} \binom{m-1-j}{r-1} e_{1,j}.\]
Regard these equations as a system of linear equations for $m-2$ variables $e_{3,0}, \ldots, e_{3,m-3}$
with the $2(m-2)$ given (but unknown) values $e_{2,j}, e_{1,j}$ for $0\leq j \leq m-3$.
Then, the matrix $A$ associated with the system is triangular, forming the Pascal's triangle
by the binomial coefficients.
Hence, $A$ has full rank and the system of equations admits a unique solution.

The rest of the proof is done by verifying the solution:
\[ e_{3,j} = 2(j+1)(n-j-2) - \sum_{i=0}^j e_{2,i} - \sum_{i=0}^j (j-i+1) e_{1,i}\]
for $0\leq j \leq m-3$.
\begin{proof}[Verification of the solution]
First, observe that
\begin{align*}
 \sum_{j=0}^{m-3}\binom{m-3-j}{r-3}\cdot 2(j+1)(n-j-2)
   & = 2\sum_{j=0}^{m-3}\binom{m-3-j}{r-3}\binom{j+1}{1} n
       -4 \sum_{j=0}^{m-3}\binom{m-3-j}{r-3}\binom{j+2}{2} \\
   & = 2\binom{m-1}{r-1}n - 4\binom{m}{r}.
\end{align*}
The last step uses a well-known identity of binomial coefficients.
See~\cite[Table 169]{gkp-cm-89}.

Secondly, we verify that
\[ \sum_{j=0}^{m-3} \left(\binom{m-3-j}{r-3} \cdot \sum_{i=0}^j e_{2,i}\right) =
   \sum_{j=0}^{m-2} \binom{m-2-j}{r-2} e_{2,j}.\]
By exchanging variables, the left-hand side is equal to
\begin{align*}
 \sum_{j=0}^{m-3} \sum_{i=0}^j \binom{m-3-j}{r-3} e_{2,i} &=
    \sum_{i=0}^{m-3} \sum_{j=i}^{m-3} \binom{m-3-j}{r-3} e_{2,i} \\
   &= \sum_{i=0}^{m-3}\left(\sum_{j=0}^{m-3-i} \binom{j}{r-3}\right) e_{2,i} \\
   &= \sum_{i=0}^{m-3} \binom{m-2-i}{r-2} e_{2,i}
   = \sum_{j=0}^{m-2} \binom{m-2-j}{r-2} e_{2,j},
\end{align*}
since $\sum_{j=0}^{a} \binom{j}{b} = \binom{a+1}{b}$ for any integers $a$ and $b$,
and $\binom{0}{r-2} = 0$ for any $r \geq 3$.

Lastly, we verify that
\[ \sum_{j=0}^{m-3} \left( \binom{m-3-j}{r-3} \cdot \sum_{i=0}^j (j-i+1)e_{1,i} \right)
   = \sum_{j=0}^{m-1} \binom{m-1-j}{r-1} e_{1,j}.\]
Similarly, we can derive that the left-hand side is equal to
\begin{align*}
 & \sum_{j=0}^{m-3} \left( \binom{m-3-j}{r-3} \cdot \sum_{i=0}^j (j-i+1)e_{1,i} \right) \\
 = &  \sum_{j=0}^{m-3} \sum_{i=0}^j \binom{m-3-j}{r-3}(j-i+1) e_{1,i} \\
 = & \sum_{i=0}^{m-3} \left(\sum_{j=i}^{m-3} (j-i+1)\binom{m-3-j}{r-3}\right) e_{1,i} \\
 = & \sum_{i=0}^{m-3} \left(\sum_{j=0}^{m-3-i} (m-2-i-j)\binom{j}{r-3}\right) e_{1,i} \\
 = & \sum_{i=0}^{m-3} \left(\sum_{j=0}^{m-2-i} \binom{m-2-i-j}{1}\binom{j}{r-3}\right) e_{1,i} \\
 = & \sum_{i=0}^{m-3} \binom{m-1-i}{r-1} e_{1,i}
 =  \sum_{j=0}^{m-1} \binom{m-1-j}{r-1} e_{1,j},
\end{align*}
since $\binom{0}{r-1} = \binom{1}{r-1} = 0$, as $r \geq 3$.
\end{proof}
This completes the proof of the theorem.
\end{proof}

Note that Theorem~\ref{thm:j-facets_convex_3d} 
reveals an 
exact equation on $e_j(S) = e_{3,j}(S) + e_{2,j}(S) + e_{1,j}(S)$
for each~$0\leq j\leq m-2$.
If $m=n$, that is, $|S_i| = 1$ for all $i\in K$,
we have $e_{1,j}(S) = e_{2,j}(S) = 0$ for every $j$,
so the equality~$e_j(S) = 2(j+1)(n-j-2)$ holds.
This exact number for the case of~$m=n$ was proved earlier
by Clarkson and Shor~\cite[Theorem 3.5]{cs-arscgII-89}.

\subsection{Euclidean color Voronoi diagrams}
Suppose that $S$ consists of $n$~points in general position in~$\Plane$,
with a given color assignment~$\kappa \colon S \to K =\{1,\ldots, m\}$,
and $\delta_s(x) = \|x-s\|_2$ is the Euclidean distance for each $s\in S$ and any $x\in \Plane$.
Consider $\CVD_k(S)$ and $\mCVD_k(S)$ for $1\leq k\leq m-1$ in this setting.

\begin{figure}[tbh]
\begin{center}
\includegraphics[width=.84\textwidth]{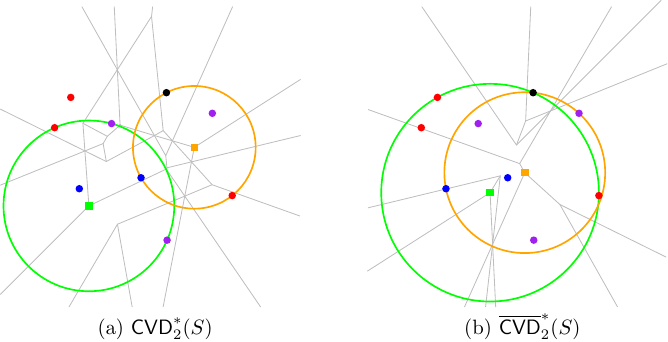}
\end{center}
\caption{Selected new vertices (small squares) in~$\CVD^*_2(S)$ and~$\mCVD^*_2(S)$
 and their corresponding circles.
 Green vertices are $2$-chromatic, while orange ones are $3$-chromatic.
 Observe that exactly one color is intersected by
 the interior~$\hat{C}$ of each circle~$C$ in (a)
 and the exterior~$\overline{C}$ of each circle~$C$ in (b).
 }
\label{fig:cvd_circles}
\end{figure}

We consider all circles through any three points in~$S$ with no regards of colors
and let $\conf(S)$ and $\overline{\conf}(S)$ be
the sets of the interiors and exteriors, respectively, of these circles.
Also, define two conflict relations $\chi \subseteq S \times \conf(S)$
and $\bar{\chi} \subseteq S \times \overline{\conf}(S)$
to be the inclusion relation.
We then consider colored configurations $\conf(S,\kappa)$ and $\overline{\conf}(S,\kappa)$
with respect to the given color assignment~$\kappa$.
Observe that
each colored configuration of weight~$j$ in~$\conf(S,\kappa)$ or in~$\overline{\conf}(S,\kappa)$
corresponds to a new vertex of~$\CVD^*_{j+1}(S)$ or of~$\mCVD^*_{j+1}(S)$, respectively,
by Lemma~\ref{lem:CVD_v} and the discussions in Section~\ref{sec:cvda}\changed{,
see} \figurename~\ref{fig:cvd_circles} illustrating the case of~$j=1$.
Therefore,
for each $1\leq c\leq 3$ and $0\leq j \leq m-c$,
we have $v_{c,j}(S) = |\conf_{c,j}(S,\kappa)|$
and $\bar{v}_{c,j}(S) = |\overline{\conf}_{c,j}(S,\kappa)|$.

Now, consider the well-known lifting that maps points $p=(x,y)$ in~$\Plane$ onto
the unit paraboloid $U = \{z = x^2 + y^2\}$ in~$\Real^{3}$:
$p=(x,y) \mapsto \lift{p} = (x, y, x^2 + y^2)$.
Let $\lift{S} = \{ \lift{s} \mid s\in S\}$ be the set of lifted colored points in~$\Real^3$.
(The horizontal plane $\{z = 0\}$ is identified \changed{as}
the original plane~$\Plane$.)
Consider colored $j$-facets in~$\lift{S}$ as in the previous section,
and recall that $e_{c,j}(\lift{S})$
denotes the number of $c$-chromatic $j$-facets in~$\lift{S}$.
We then observe the following.
\begin{lemma} \label{lem:v-e} 
 For $1\leq c\leq 3$ and $0\leq j\leq m-c$, we have
  $v_{c,j}(S) + \bar{v}_{c,j}(S) = e_{c,j}(\lift{S})$.
\end{lemma}
\begin{proof}
Regard the $z$-direction in~$\Real^3$ as the \emph{vertical} direction
and call each $j$-facet in~$\lift{S}$ \emph{downward} or \emph{upward}
according to the half-space it chooses.
Since points in~$S$ are in general position,
no three points in~$\lift{S}$ lie on a common vertical plane.

Consider any downward $c$-chromatic $j$-facet in~$\lift{S}$ and
its corresponding half-space $h^-$.
Then, the orthogonal projection of the intersection $h^- \cap U$ onto $\Real^2 = \{z=0\}$
is the interior of a circle $C$ such that
there are three points in $S$ from $c$ different colors lying on $C$
and the interior of $C$ intersects exactly $j$ colors.
That is, the interior of $C$ is a member of the set $\conf_{c,j}(S, \kappa)$.
Since the lifting is bijective,
we can establish a one-to-one correspondence between $\conf_{c,j}(S,\kappa)$
and the set of \emph{downward} $c$-chromatic $j$-facets in~$\lift{S}$.

Next, consider any upward $c$-chromatic $j$-facet in~$\lift{S}$ and
its corresponding half-space $h^+$.
Then, the orthogonal projection of the intersection $h^+\cap U$ onto $\Real^2 = \{z=0\}$
is the exterior of a circle $C$ such that
there are three points in $S$ from $c$ different colors lying on $C$
and the exterior of $C$ intersects exactly $j$ colors.
This way, there is a one-to-one correspondence between $\overline{\conf}_{c,j}(S,\kappa)$
and the set of \emph{upward} $c$-chromatic $j$-facets in~$\lift{S}$.
Hence, the lemma follows.
\end{proof}

Since $\lift{S}$ is in convex and general position,
Theorem~\ref{thm:j-facets_convex_3d},
together with Lemmas~\ref{lem:CVD_v} and~\ref{lem:v-e}, implies an exact equation
on the number of vertices in $\CVD_k(S)$ and $\mCVD_k(S)$.
\begin{theorem} \label{thm:kcvd_Euclidean} 
 Let $S$ be a set of $n$~points with $m$~colors in general position in the Euclidean plane,
 and $1\leq k \leq m-1$.
 The total number of vertices in $\CVD_k(S)$ and $\mCVD_k(S)$ is exactly
  \[ 4k(n-k)-2n - 2\sum_{i=0}^{k-2} e_{2,i}(\lift{S}) - \sum_{i=0}^{k-1}(2k-2i-1)e_{1,i}(\lift{S}).\]
\end{theorem}
\begin{proof}
For each $1\leq k \leq m-1$, let $v_k$ be the number of vertices in $\CVD_k(S)$
and $\bar{v}_k$ be the number of vertices in $\mCVD_k(S)$.
By Lemmas~\ref{lem:CVD_v} and~\ref{lem:v-e}, it holds that
\[ v_k + \bar{v}_k = e_{3,k-1}(\lift{S}) + e_{3,k-2}(\lift{S}) + e_{2,k-1}(\lift{S}), \]
where $e_{c,j} = 0$ for $j< 0$.

It is easy to see that
 \[ v_1 + \bar{v}_1 = e_{3,0}(\lift{S}) + e_{2,0}(\lift{S}) = 2n-4 - e_{1,0}(\lift{S}),\]
by Theorem~\ref{thm:j-facets_convex_3d}.

For each $2\leq k \leq m-1$,
Theorem~\ref{thm:j-facets_convex_3d} implies that
\begin{align*}
 v_k + \bar{v}_k
  &= 2k(n-k-1) + 2(k-1)(n-k) - 2\sum_{i=0}^{k-2} e_{2,i}(\lift{S}) - \sum_{i=0}^{k-1}(2k-2i-1)e_{1,i}(\lift{S}) \\
  &= 4k(n-k) - 2n - 2\sum_{i=0}^{k-2} e_{2,i}(\lift{S}) - \sum_{i=0}^{k-1}(2k-2i-1)e_{1,i}(\lift{S}).
\end{align*}
Hence, for every $1\leq k \leq m-1$,
the claimed equation holds.
\end{proof}

Theorem~\ref{thm:kcvd_Euclidean} implies
the $O(k(n-k))$ bound on the complexity of $\CVD^*_k(S)$ and $\mCVD^*_k(S)$
by Lemma~\ref{lem:CVD_v}.
An interesting special case of the above result is the following.
\begin{corollary}
 Given a set $S$ of $n$ colored points with $m$ colors in the Euclidean plane,
 the complexity of both~$\FCVD^*(S)$ and~$\HVD^*(S)$ is bounded by $O(m(n-m+1))$.
\end{corollary}

\section{Color Voronoi diagrams under general distance functions}
\label{sec:kcvd_general}

We extend our results for the Euclidean case to general distance functions.
We continue the discussions from Section~\ref{sec:cvda},
so $S$ is a set of $n$~sites, colored with $m$~colors from~$K$ by a color assignment~$\kappa$, 
and
the functions~$\delta_s \colon \Plane \to \Real$ for~$s\in S$
are in general position.

Notice that the vertices of the arrangements~$\arr_\Gamma$ and $\overline{\arr}_\Gamma$
form two set systems that fit in our colorful framework.
More precisely,
let $\conf(S)$ be the set of vertices of the arrangement of $n$~surfaces in~$\Gamma$,
and consider two conflict relations $\chi, \bar{\chi} \subseteq S\times \conf(S)$
such that $(s, v)\in \chi$ if $v \in \conf(S)$ lies above surface~$\gamma_s \in \Gamma$
and $(s, v) \in \bar{\chi}$ if $v$ lies below $\gamma_s$.
Based on two CS-structures $(S, \conf(S), \chi)$ and $(S, \conf(S), \bar{\chi})$,
we consider their colored configurations with respect to~$\kappa$,
denoted by $\conf(S, \kappa)$ and $\overline{\conf}(S, \kappa)$, respectively.
By this construction, we have 
$v_{c,j}(S) = |\conf_{c,j}(S, \kappa)|$
and $\bar{v}_{c,j}(S) = |\overline{\conf}_{c,j}(S, \kappa)|$,
counting $c$-chromatic weight-$j$ colored configurations in~$\conf(S,\kappa)$ and 
in~$\overline{\conf}(S,\kappa)$, respectively,
and, simultaneously, counting new $c$-chromatic vertices in~$\CVD^*_{j+1}(S)$
and in $\mCVD^*_{j+1}(S)$ by Lemma~\ref{lem:CVD_v}.

For each $S'\subseteq S$,
let $v_0(S')$ and $u_0(S')$ denote the numbers of vertices
and unbounded edges\footnote{
Hereafter, by counting unbounded edges, we mean counting vertices at infinity.
So, if an unbounded edge separates the plane~$\Plane$, then it is counted twice.
}, 
respectively, in~$\VD(S')$;
let $\bar{v}_0(S')$ and $\bar{u}_0(S')$ denote the numbers of vertices
and unbounded edges, respectively, in $\FVD(S')$.
\changed{We
consider} the following conditions.

\begin{itemize} 
 \item[\textbf{V1}] $v_0(S') = 2|S'| - 2 - u_0(S')$ for any~$S' \subseteq S$.
 \item[\textbf{V2}] $\bar{v}_0(S') = \bar{u}_0(S') - 2$ for any~$S' \subseteq S$.
\end{itemize}

Note that if
every region in $\VD(S')$ is nonempty and simply connected,
then Euler's formula and the general position assumption
imply condition~\textbf{V1}\changed{. If}
$\FVD(S')$ forms a tree, then every face of~$\FVD(S')$ is unbounded and thus
\changed{condition~\textbf{V2} holds} by Euler's formula.

In addition, for $c\in \{1,2\}$ and $j \geq 0$,
let $u_{c,j} = u_{c,j}(S)$
be the number of $c$-chromatic unbounded edges in $\mathcal{A}_\Gamma$
that lie \emph{above} exactly $j$~surfaces in $\{E_i\}_{i\in K}$, and
$\bar{u}_{c,j} = \bar{u}_{c,j}(S)$ be the number of $c$-chromatic unbounded edges in
$\overline{\mathcal{A}}_\Gamma$
that lie \emph{below} exactly $j$~surfaces in $\{\overline{E}_i\}_{i\in K}$.
From the
\changed{discussion} in Section~\ref{sec:cvda},
observe that $u_{c,j}$ and $\bar{u}_{c,j}$ are equal to
the numbers of new $c$-chromatic unbounded edges
in $\CVD^*_{j+1}(S)$ and in $\mCVD^*_{j+1}(S)$, respectively.
Further, as for~$v_{c,j}$ and $\bar{v}_{c,j}$, note that
$u_{c,j}$ and $\bar{u}_{c,j}$ indeed count $c$-chromatic weight-$j$ colored configurations  
based on two CS-structures for unbounded edges in
the arrangement of $n$~surfaces in~$\Gamma$.
Hence, assuming~\textbf{V1} and~\textbf{V2},   
Lemma~\ref{lem:CS_general_ub} implies:
for any subset $R \subseteq K$,
 \[ \sum_{c=1}^3 v_{c,0}(S_R) = 2|S_R| - 2 - \sum_{c=1}^2 u_{c,0}(S_R)
  \quad \text{ and } \quad
    \sum_{c=1}^3 \bar{v}_{c,0}(S_R) = \sum_{c=1}^2 \bar{u}_{c,0}(S_R) - 2,\] 
since $\CVD^*_1(S_R) = \VD(S_R)$ and $\mCVD^*_1(S_R) = \FVD(S_R)$.
Combining these equations and the others obtained by Lemma~\ref{lem:CS_general_lb},
we have two systems of linear equations that can be solved
in a similar way as done in~Theorem~\ref{thm:j-facets_convex_3d}.
For $0\leq j \leq m-1$, define
\begin{align*}
  V_j  := v_{3,j} + \sum_{i=0}^j (v_{2,i} + (j-i+1)v_{1,i}), \quad \text{\phantom{ and }}\quad
  & U_j := \sum_{i=0}^j (u_{2,i} + (j-i+1)u_{1,i}),\\
 \overline{V}_j  := \bar{v}_{3,j} + \sum_{i=0}^j (\bar{v}_{2,i} + (j-i+1)\bar{v}_{1,i}),
  \quad \text{ and }\quad &
 \overline{U}_j  := \sum_{i=0}^j (\bar{u}_{2,i} + (j-i+1)\bar{u}_{1,i}).
 \end{align*}

\begin{lemma} \label{lem:kcvd_general_uv} 
 With the above notations, let $0\leq j\leq m-2$.
 Condition~\textbf{V1} implies
 $V_j + U_j = (j+1)(2n-j-2)$;
 condition~\textbf{V2} implies
  $\overline{V}_j - \overline{U}_j = -(j+1)(j+2)$.
\end{lemma}
\begin{proof}
By the general position assumption on the functions $\delta_s$ for $s\in S$,
Lemma~\ref{lem:CS_general_lb} implies:
for $1\leq r\leq m$ and a random subset $R\subseteq K$ of $r$~colors,
 \[ \binom{m}{r} \E[v_{c,0}(S_R)] = \sum_{j=1}^{m-c} v_{c,j}\binom{m-c-j}{r-c}
  \text{ and }
    \binom{m}{r} \E[\bar{v}_{c,0}(S_R)] = \sum_{j=0}^{m-c} \bar{v}_{c,j}\binom{m-c-j}{r-c},
 \]
for each $c \in \{1,2,3\}$,
and
 \[ \binom{m}{r} \E[u_{c,0}(S_R)] = \sum_{j=1}^{m-c} u_{c,j}\binom{m-c-j}{r-c}
  \text{ and }
    \binom{m}{r} \E[\bar{u}_{c,0}(S_R)] = \sum_{j=0}^{m-c} \bar{u}_{c,j}\binom{m-c-j}{r-c},
 \]
for each $c\in\{1,2\}$.
Hence, on one hand, we have
 \[ \binom{m}{r}\left(\sum_{c=1}^3 \E[v_{c,0}(S_R)] + \sum_{c=1}^2 \E[u_{c,0}(S_R)]\right)
   = \sum_{c=1}^3 \sum_{j=0}^{m-c} (v_{c,j}+u_{c,j})\binom{m-c-j}{r-c}\]
and
 \[ \binom{m}{r}\left(\sum_{c=1}^3 \E[\bar{v}_{c,0}(S_R)] - \sum_{c=1}^2 \E[\bar{u}_{c,0}(S_R)]\right)
   = \sum_{c=1}^3 \sum_{j=0}^{m-c} (\bar{v}_{c,j}-\bar{u}_{c,j})\binom{m-c-j}{r-c},\]
where we define $u_{3,j} = \bar{u}_{3,j} = 0$ for all $j$.

On the other hand, Lemma~\ref{lem:CS_general_ub}, together with the above discussions, implies that
\begin{align*}
 \binom{m}{r}\left(\sum_{c=1}^3 \E[v_{c,0}(S_R)] + \sum_{c=1}^2 \E[u_{c,0}(S_R)]\right)
 &= \sum_{R'\subseteq K, |R'| = r} \sum_{c=1}^3 (v_{c,0}(S_{R'}) + u_{c,0}(S_{R'})) \\
 &= \sum_{R'} (2|S_{R'}| - 2) \\
 & = 2\binom{m-1}{r-1}n - 2\binom{m}{r},
\end{align*}
for any $2\leq r\leq m$, by Lemma~\ref{lem:sums},
if condition~\textbf{V1} is fulfilled.
Similarly, we also have
\[
 \binom{m}{r}\left(\sum_{c=1}^3 \E[\bar{v}_{c,0}(S_R)] - \sum_{c=1}^2 \E[\bar{u}_{c,0}(S_R)]\right)
 = - 2\binom{m}{r},
\]
for any $2\leq r\leq m$, if condition~\textbf{V2} is fulfilled.

Hence, assuming condition~\textbf{V1}, we have $m-1$~linear equations:
\[
 \sum_{c=1}^3 \sum_{j=0}^{m-c} \binom{m-c-j}{r-c} (v_{c,j} + u_{c,j})
   = 2\binom{m-1}{r-1}n - 2\binom{m}{r},
\]
for $2\leq r\leq m$.
The equation for $r=2$ is written as
\begin{align*}
 \sum_{c=1}^3 \sum_{i=0}^{m-c} \binom{m-c-i}{2-c} (v_{c,i} + u_{c,i})
  & = \sum_{i=0}^{m-2} (v_{2,i} + u_{2,i}) + \sum_{i=0}^{m-1}(m-1-i)(v_{1,i}+u_{1,i}) \\
  & = V_{m-2} + U_{m-2} = 2\binom{m-1}{1}n - 2\binom{m}{2}\\
  & = (m-1)(2n-m),
\end{align*}
as claimed, since $v_{3,m-2} = 0$.
Now, consider the other $m-2$ equations for $3\leq r\leq m$,
forming a system of linear equations with $m-2$ variables $v_{3,0}, \ldots, v_{3, m-3}$.
This system is associated with the same matrix~$A$ as in Theorem~\ref{thm:j-facets_convex_3d},
so it admits a unique solution:
\[
 v_{3,j} = (j+1)(2n-j-2) - \sum_{i=0}^j(v_{2,i} + (j-i+1)v_{1,i}) 
    - \sum_{i=0}^j(u_{2,i} + (j-i+1)u_{1,i}),
\]
for $0\leq j\leq m-3$,
which can be easily verified as done in the proof of Theorem~\ref{thm:j-facets_convex_3d}.
Thus, we conclude the claimed equations
\[
  V_j + U_j = (j+1)(2n-j-2),
\]
for $0\leq j\leq m-2$.

Analogously, assuming condition~\textbf{V2}, the above discussion results in
the following $m-1$ linear equations:
\[
 \sum_{c=1}^3 \sum_{j=0}^{m-c} \binom{m-c-j}{r-c} (\bar{v}_{c,j} - \bar{u}_{c,j})
   = - 2\binom{m}{r},
\]
for $2\leq r\leq m$.
The equation for $r=2$ is written as
\begin{align*}
 \sum_{c=1}^3 \sum_{i=0}^{m-c} \binom{m-c-i}{2-c} (\bar{v}_{c,i} - \bar{u}_{c,i})
  & = \sum_{i=0}^{m-2} (\bar{v}_{2,i} - \bar{u}_{2,i}) 
      + \sum_{i=0}^{m-1}(m-1-i)(\bar{v}_{1,i}+\bar{u}_{1,i}) \\
  & = \overline{V}_{m-2} - \overline{U}_{m-2} = - 2\binom{m}{2}\\
  & = (m-1)m,
\end{align*}
as claimed, since $\bar{v}_{3,m-2} = 0$.
The other equations for $3\leq r\leq m$ form
a system of $m-2$ linear equations with $m-2$ variables $\bar{v}_{3,0}, \ldots, \bar{v}_{3,m-2}$
whose coefficients are exactly the same as above.
Its unique solution is
\[
 \bar{v}_{3,j} = -(j+1)(j+2) - \sum_{i=0}^j(\bar{v}_{2,i} + (j-i+1)\bar{v}_{1,i}) 
    + \sum_{i=0}^j(\bar{u}_{2,i} + (j-i+1)\bar{u}_{1,i}),
\]
for $0\leq j\leq m-3$,
which can be written as
\[
  \overline{V}_j - \overline{U}_j = -(j+1)(j+2).
\]
Hence, we conclude the lemma.
\end{proof}

\begin{theorem} \label{thm:kcvd_general} 
  Given $S$ and $\{\delta_s\}_{s\in S}$ in general position as above, 
  let $1\leq k\leq m-1$ be an integer.
 If condition~\textbf{V1} is true,
 then the number of vertices in~$\CVD_k(S)$ is exactly
  \[ 2k(2n-k)-2n - 2\sum_{i=0}^{k-2} v_{2,i}(S) - \sum_{i=0}^{k-1} (2k-2i-1)v_{1,i}(S)
      - U_{k-1} - U_{k-2};\]
 if condition~\textbf{V2} is true,
 the number of vertices in~$\mCVD_k(S)$ is exactly
  \[ \overline{U}_{k-1} + \overline{U}_{k-2} - 2k^2
      - 2\sum_{i=0}^{k-2}\bar{v}_{2,i}(S)-\sum_{i=0}^{k-1}(2k-2i-1)\bar{v}_{1,i}(S).\]
  \end{theorem}
\begin{proof}
By Lemma~\ref{lem:CVD_v}, the number of vertices in $\CVD_k(S)$ is exactly
 $v_{3,k-1} + v_{3,k-2} + v_{2,k-1}$
and the number of vertices in $\mCVD_k(S)$ is exactly
 $\bar{v}_{3,k-1} + \bar{v}_{3,k-2} + \bar{v}_{2,k-1}$.
Plugging the equations shown in Lemma~\ref{lem:kcvd_general_uv}
results in the claimed exact quantities.
\end{proof}

Remark that condition~\textbf{V1} already implies
the asymptotic complexity~$O(k(n-k))$ of~$\CVD_k(S)$ for~$k \leq \frac{n}{2}$
as $2n-k \leq 3(n-k)$,
while we have $O(kn)$ for~$k > \frac{n}{2}$.
Further, to show the $O(k(n-k))$ bound
for any value of $k$ for both $\CVD_k(S)$ and $\mCVD_k(S)$,
it suffices to show that $U_{j} \geq (j+1)(j+2) - o(j^2)$
and $\overline{U}_{j} \leq (j+1)(2n - j -2) + o(j^2)$.
In this way, Theorem~\ref{thm:kcvd_general} reduces the problem of bounding the complexity of
higher-order color Voronoi diagrams to that of bounding the number of their unbounded edges.
Also, note that Lemma~\ref{lem:kcvd_general_uv} implies
$ U_j \leq (j+1)(2n-j-2) \quad \text{and} \quad \overline{U}_j \geq (j+1)(j+2)$,
if conditions~\textbf{V1} and~\textbf{V2} hold.

Remark also that if $\VD(S')$ and $\FVD(S')$ \changed{fall under the
umbrella  of abstract Voronoi diagrams,}
then conditions~\textbf{V1} and~\textbf{V2}
\changed{hold}~\cite{k-cavd-89,mmr-fsavd-01},
so Theorem~\ref{thm:kcvd_general} implies:
\begin{corollary} \label{coro:kcvd_AVD}
 Suppose that $S$ and $\{\delta_s\}_{s\in S}$ 
 imply a bisector system that satisfies
 the conditions of \changed{abstract Voronoi diagrams~\cite{k-cavd-89}.}
 Then, the complexity of~$\CVD_k(S)$ and~$\CVD^*_k(S)$ is
 $O(k(n-k))$ for $1\leq k\leq \lfloor \frac{n}{2} \rfloor$ 
 and $O(kn)$ for $\lfloor \frac{n}{2} \rfloor + 1 \leq k \leq m$.
\end{corollary}

The quantities $U_j$ and $\overline{U}_j$, related to the number of unbounded edges,
often turn out to be equal;
the very typical example is the Euclidean case for point sites~$S$,
where the equality $u_{c,j}(S) = \bar{u}_{c,j}(S) = e_{c,j}(S)$ holds\changed{.}
This inspires us to consider the following third condition:

\begin{itemize}
 \item[\textbf{V3}] $u_0(S') = \bar{u}_0(S')$ for any~$S'\subseteq S$.
\end{itemize}

\begin{lemma} \label{lem:V3} 
 Condition~\textbf{V3} implies $U_j = \overline{U}_j$
 for any $0\leq j\leq m-1$.
\end{lemma}

\begin{proof}
Condition~\textbf{V3} implies that
 \[u_{2,0}(S_{R'}) + u_{1,0}(S_{R'}) = \bar{u}_{2,0}(S_{R'})+\bar{u}_{1,0}(S_{R'})\]
for any subset $R' \subseteq K$.
Now, let $r$ be any integer with $1\leq r\leq m$ and
$R \subseteq K$ be a random $r$-color subset.
Then, Lemma~\ref{lem:CS_general_lb} implies that
\begin{align*}
0&=\E[u_{2,0}(S_R) + u_{1,0}(S_R) - \bar{u}_{2,0}(S_R)-\bar{u}_{1,0}(S_R)] \\
    & = \E[u_{2,0}(S_R)] + \E[u_{1,0}(S_R)] - \E[\bar{u}_{2,0}(S_R)] - \E[\bar{u}_{1,0}(S_R)]\\
    & = \sum_{c=1}^2 \sum_{i=0}^{m-c} (u_{c,i} - \bar{u}_{c,i}) \binom{m-c-i}{r-c}.
\end{align*}
Letting $\zeta_{c,j} := u_{c,j} - \bar{u}_{c,j}$,
we have $m$ equations for $1\leq r\leq m$:
\[ \sum_{i=0}^{m-2} \binom{m-2-i}{r-2} \zeta_{2,i} + \sum_{i=0}^{m-1} \binom{m-1-i}{r-1} \zeta_{1,i}
  = 0.\]

Regarding $\zeta_{2,0}, \ldots, \zeta_{2,m-2}$ as $m-1$ variables,
consider the system formed by the $m-1$ linear equations for $2\leq r\leq m$.
As done in Theorem~\ref{thm:j-facets_convex_3d} and Lemma~\ref{lem:kcvd_general_uv},
this system admits a unique solution:
\[ \zeta_{2,i} + (\zeta_{1,0} + \cdots + \zeta_{1,i}) = 0,\]
for $0\leq i\leq m-2$.
Observe that this also holds for $i=m-1$ from the above equation for $r=1$, 
that is, $\zeta_{2,m-1} + \sum_{i=0}^{m-1} \zeta_{1,i}= 0$,
since $u_{2,m-1} = \bar{u}_{2,m-1} = 0$. 

To conclude the lemma, for each $0\leq j\leq m-1$, 
we sum up the solution over $0\leq i \leq j$,
which results in
\[ \sum_{i=0}^j (u_{2,i} - \bar{u}_{2,i}) + \sum_{i=0}^j (j-i+1)(u_{1,i} - \bar{u}_{1,i}) 
  = U_j - \overline{U}_j =  0.\]
This completes the proof.
\end{proof}

Assuming conditions~\textbf{V1}--\textbf{V3},
we obtain the same exact number as in Theorem~\ref{thm:kcvd_Euclidean}.
\begin{theorem} \label{thm:kcvd_general_2n-4} 
 Given $S$ and $\{\delta_s\}_{s\in S}$ in general position as above,
 if conditions~\textbf{V1}--\textbf{V3} hold,
 then the total number of vertices in~$\CVD_k(S)$ and $\mCVD_k(S)$
 for $1\leq k\leq m-1$
 is exactly
  \[ 4k(n-k)-2n - 2\sum_{i=0}^{k-2} (v_{2,i}(S)+\bar{v}_{2,i}(S))
     - \sum_{i=0}^{k-1} (2k-2i-1)(v_{1,i}(S) + \bar{v}_{1,i}(S)).\]
\end{theorem}
\begin{proof}
From conditions~\textbf{V1} and~\textbf{V2}, we have
 \[ V_j + \overline{V}_j + U_j - \overline{U}_j = 2(j+1)(n-j-2)\]
for any $0\leq j\leq m-2$, by Lemma~\ref{lem:kcvd_general_uv}.
Since $U_j = \overline{U}_j$ by Lemma~\ref{lem:V3} with condition~\textbf{V3}, 
we indeed have
 \[ V_j + \overline{V}_j = 2(j+1)(n-j-2),\]
which is exactly the same equation we obtain in the Euclidean case (Theorem~\ref{thm:kcvd_Euclidean}).
Hence, the claimed exact number follows.
\end{proof}

Below, we discuss \changed{some}
specific cases of functions~$\delta_s$
for a set~$S$ of points in the plane~$\Plane$,
in which new results are derived
by applying Theorems~\ref{thm:kcvd_general} and~\ref{thm:kcvd_general_2n-4}.

\subparagraph{Convex distance functions.}
From now on, suppose $S$ consists of $n$~colored points in~$\Plane$.
Let $B \subset \Plane$ be any convex and compact body 
whose interior contains the origin.
\changed{Define} $\delta_s(x)=\|x-s\|_B$ for point~$s\in S$ to be
the \emph{convex distance} from~$x\in\Plane$ to~$s$ based on~$B$~\cite{cd-vdcdf-85,m-bvdcdf-00}.
Since Voronoi diagrams of point sites under a convex distance function
\changed{fall under}
the model of abstract Voronoi diagrams~\cite{k-cavd-89,mmr-fsavd-01},
\changed{conditions~\textbf{V1} and~\textbf{V2} hold.} 

Condition~\textbf{V3}, however, is not guaranteed in general;
a popular example is the $L_1$ or $L_\infty$ metric,
under which $\VD(S')$ may have $\Theta(|S'|)$ parallel unbounded edges
while $\FVD(S')$ has at most four \changed{unbounded edges.}
In the following, we first assume \changed{that $B$ is \emph{smooth}}
that is, there is a unique line
tangent to~$B$ at each point on its boundary~\cite{m-bvdcdf-00}.
We then make the following observation, stronger than condition~\textbf{V3},
which
\changed{has been} known for the Euclidean metric
even in higher dimensions~\cite{BPS23}.

\begin{lemma} \label{lem:V3-convex_distance} 
 Given $S$ and $\delta_s$ for $s\in S$ as above,
 suppose $B$ is smooth.
 For $c\in\{1,2\}$ and $0\leq j\leq m-1$, 
 we have $u_{c,j}(S) = \bar{u}_{c,j}(S) = e_{c,j}(S)$,
 the number of $c$-chromatic $j$-facets in~$S$.
\end{lemma}
\begin{proof}
Pick any $c$-chromatic $j$-facet~$\sigma$ in~$S$.
Let $s_1,s_2\in S$ be two points that define~$\sigma$,
$\ell$ the line through $s_1$ and $s_2$,
and $\sigma^+$ the open half-plane bounded by~$\ell$ that is chosen by~$\sigma$.
Thus,
$\sigma^+$ intersects exactly $j$~colors from~$K\setminus \{\kappa(s_1), \kappa(s_2)\}$.

\begin{figure}[hb]
\begin{center}
\includegraphics[width=.9\textwidth]{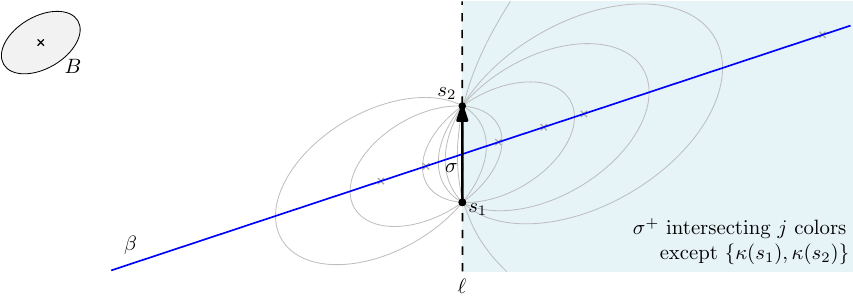}
\end{center}
\caption{
 Illustration to Lemma~\ref{lem:V3-convex_distance}.
 }
\label{fig:V3_convex_distance}
\end{figure}

Now, consider all scaled and translated copies of~$B$ that go through both~$s_1$ and~$s_2$.
The \emph{center} of such a copy~$B'$ of~$B$ is the origin translated 
by the same translation vector of~$B'$.
Let $\beta$ be the locus of the centers of all these copies,
which forms an unbounded curve \changed{splitting~$\Plane$,}
often called the \emph{bisector} between 
$s_1$ and $s_2$~\cite{cd-vdcdf-85,m-bvdcdf-00}. 
For $p \in \beta$, let $B(p)$ be the scaled and translated copy of~$B$
such that its center is~$p$ and $B(p)$ goes through $s_1$ and $s_2$.
Note that the boundary of~$B(p)$ tends to be~$\ell$
as $p$ goes to the point at infinity in either direction along~$\beta$,
since $B$ is convex and smooth.
Let $\hat{B}(p)$ and $\overline{B}(p)$ for $p\in\beta$ be the interior and the exterior of~$B(p)$,
respectively, excluding the boundary of~$B(p)$.
Then, observe that $\hat{B}(p)$ tends to be~$\sigma^+$ as $p$ goes in one direction along~$\beta$,
while, if $p$ goes in the other direction, $\overline{B}(p)$ tends to be~$\sigma^+$ as well,
since $B$ is convex and smooth. 
See \figurename~\ref{fig:V3_convex_distance} for an illustration,
in which $B$ is an ellipse (so, $\beta$ appears to be a line)
and $\sigma^+$ is shaded in lightblue.
From the discussions in Section~\ref{sec:cvda} 
and the analog of the Euclidean case in Section~\ref{sec:framework}.2, 
the endpoint of~$\beta$ (at infinity) in the first direction is a vertex at infinity of~$\CVD^*_{j+1}(S)$
incident to a new $c$-chromatic unbounded edge, which we denote by~$\eta_{c,j}(\sigma)$;
the other endpoint of~$\beta$ in the other direction is 
a vertex at infinity of~$\mCVD^*_{j+1}(S)$
incident to a new $c$-chromatic unbounded edge, which we denote by~$\bar{\eta}_{c,j}(\sigma)$.
Note that both $\eta_{c,j}(\sigma)$ and $\bar{\eta}_{c,j}(\sigma)$ are defined by
the same pair~$(s_1, s_2)$ of sites as $\sigma$.

Observe that the above argument also shows that both $\eta$ and $\bar{\eta}$ 
are bijective.
For any new $c$-chromatic unbounded edge~$\beta'$ of~$\CVD^*_{j+1}(S)$,
consider $\hat{B}(p)$ for $p\in \beta'$.
Then, the limit of~$\hat{B}(p)$, as $p$ goes to the vertex at infinity incident to~$\beta'$,
tends to be an open half-plane that intersects exactly $j$~colors,
which determines a $c$-chromatic $j$-facet in~$S$.
Hence, $\eta_{c,j}$ is a bijection.
Analogously, $\bar{\eta}_{c,j}$ is a bijection as well. 
This results in the equation $u_{c,j}(S) = \bar{u}_{c,j}(S) = e_{c,j}(S)$
for any $c$ and $j$.
\end{proof}
By Lemma~\ref{lem:V3-convex_distance},
Theorem~\ref{thm:kcvd_general_2n-4} implies 
the same upper bound~$4k(n-k)-2n$ on the total number of vertices
in \changed{both} $\CVD_k(S)$ and $\mCVD_k(S)$ under any smooth convex distance function.

We then relax the smoothness of~$B$ by a limit argument,
so let $B$ be any convex and compact body.
Consider a sequence of smooth and convex bodies~$B_0, B_1, \ldots$
that converges to~$B$.
Obviously, there
\changed{exists} $\hat{B} = B_i$ sufficiently close to~$B$ so
that\changed{: (a)}
the functions~$\hat{\delta}_s(x) = \|x-s\|_{\hat{B}}$~for~$s\in S$
 are in general position (see Section~\ref{sec:cvda}), and \changed{(b)}
 for any scaled and translated copy $B_{pqr}$ of~$B$
 having three points $p,q,r\in S$ on its boundary,
 there is also a scaled and translated copy $\hat{B}_{pqr}$
 of~$\hat{B}$,
 \changed{whose boundary goes through $p,q,r$}
 and the separation of~$S$ by~$\hat{B}_{pqr}$ is the same as that by~$B_{pqr}$.

This implies that the vertices in $\CVD_k(S)$ and in $\mCVD_k(S)$ under the convex distance function
based on~$B$ are preserved \changed{in their counterpart diagrams}
under the convex distance function based \changed{on~$\hat{B}$.}
Furthermore, the general position assumption can be relaxed,
as it does not decrease the number of vertices in the diagrams.
Hence, we conclude:
\begin{corollary} \label{coro:kcvd_convex_distance}
 Let $S$ be a set of $n$ colored points in~$\Plane$ with $m$ colors.
 For $1\leq k\leq m-1$,
 the total number of vertices in $\CVD_k(S)$ and $\mCVD_k(S)$ 
 under any $L_p$ metric for $1\leq p \leq \infty$ or any convex distance function
 is at most $4k(n-k)-2n$.
\end{corollary}

It is worth noting that Lemmas~\ref{lem:kcvd_general_uv} and~\ref{lem:V3-convex_distance} 
imply bounds for colored $j$-facets in~$\Plane$.
\begin{corollary} \label{coro:j-facets_2d} 
 Let $S$ be a set of $n$ colored points in~$\Plane$ with $m$ colors.
 For $0\leq k \leq m-2$,
 \[ (k+1)(k+2) \leq \sum_{j=0}^k e_{2,j}(S) + \sum_{j=0}^k (k-j+1)e_{1,j}(S) \leq (k+1)(2n-k-2). \]
\end{corollary}
\begin{proof}
Let $B$ be any smooth and convex body in~$\Plane$,
and consider the color Voronoi diagrams $\CVD_k(S)$ and $\mCVD_k(S)$
under the convex distance function based on~$B$.
In this setting, we know
by Lemma~\ref{lem:V3-convex_distance} that conditions~\textbf{V1}--\textbf{V3} hold
and $u_{c,j} = \bar{u}_{c,j} = e_{c,j}(S)$.
So, we have
 \[ U_k = \overline{U}_k = \sum_{j=0}^k e_{2,j}(S) + \sum_{j=0}^k (k-j+1)e_{1,j}(S),\]
for $0\leq k\leq m-2$ by Lemma~\ref{lem:V3}.
From Lemma~\ref{lem:kcvd_general_uv},
we have $V_k \leq V_k + \overline{V}_k = 2(k+1)(n-k-2)$ and $V_k + U_k = (k+1)(2n-k-2)$,
which implies the claimed lower bound
 \[ U_k \geq (k+1)(k+2)\]
for any $0\leq k\leq m-2$.

For the upper bound,
Lemma~\ref{lem:kcvd_general_uv} implies 
$\overline{V}_k \leq V_k + \overline{V}_k = 2(k+1)(n-k-2)$
and thus
\[ \overline{U}_k = \overline{V}_k + (k+1)(k+2) \leq (k+1)(2n-k-2)\]
for any $0\leq k\leq m-2$.
\end{proof}

\subparagraph{More on polygonal convex distance functions.}
First, we consider the $L_\infty$ metric,
so $B$ is the unit square centered at the origin.
Liu, Papadopoulou, and Lee~\cite{lpl-knnvdr-15} proved 
an upper bound of $O((n-k)^2)$ for ordinary order-$k$ Voronoi diagrams
of $n$~points under the $L_\infty$ metric using the \emph{Hanan grid}.
An analogous argument can also be applied to our color diagrams.
\begin{lemma}\label{lem:kcvd_L1_grid} 
 Under the $L_\infty$ metric, the number of vertices in $\CVD_k(S)$ 
 is at most $4(n-k)^2$.
\end{lemma}
\begin{proof}
The \emph{Hanan grid} $G=G(S)$ of $S$ is a grid constructed by drawing two lines, vertical and horizontal,
through each point in $S$~\cite{lpl-knnvdr-15}.
Consider an $L_\infty$-circle~$\square$ corresponding to a \emph{new} vertex of $\CVD_k(S)$
under the $L_\infty$ metric.
Note that the interior of~$\square$ intersects exactly $k-1$ colors,
so $\square$ includes at least $k-1$~points of~$S$ in its interior.
Also, exactly three points of~$S$ lie on~$\square$ and
two adjacent corners of~$\square$ lie on grid points of $G$ by the general position assumption
(See Lemma~7 in~\cite{lpl-knnvdr-15}),
so either the top-left corner or the bottom-right corner of~$\square$ is a grid point,
but not both.

Assume that the top-left corner $y$ of~$\square$ is a grid point of $G$.
If $y$ is the intersection of the $a$-th vertical line of~$G$ from the left
and the $b$-th horizontal line from above,
then we have $a \leq n-k-1$ and $b \leq n-k-1$
since we need at least $k+1$~points below and to the right of~$y$
to form such a square~$\square$.
Since each grid point can serve at most once as top-left corner of such a square,
there are at most $(n-k-1)^2$ new vertices of~$\CVD_k(S)$ 
whose corresponding $L_\infty$-circles have their top-left corners lie at grid points.
The same argument also applies to those squares whose bottom-right corner lies at a grid point.
Therefore, we conclude that
\[ v_{3,k-1} + v_{2,k-1} + v_{1,k-1} \leq 2(n-k-1)^2.\]

By Lemma~\ref{lem:CVD_v}, the number of vertices of~$\CVD_k(S)$ is
\[ v_{3,k-1} + v_{3,k-2} + v_{2,k-1} \leq 2(n-k-1)^2 + 2(n-k)^2 \leq 4(n-k)^2,\]
as claimed.
\end{proof}

Hence, the complexity of~$\CVD_k(S)$ under the $L_\infty$ metric is 
$O(\min\{k(n-k), (n-k)^2\})$.
For the maximal counterpart $\mCVD_k(S)$,
we prove the following.
\begin{lemma} \label{lem:kcvd_L1_barU} 
 Under the $L_\infty$ metric, for any $0\leq j\leq m-2$, we have
  $\overline{U}_j \leq 2(j+1)(j+2)$.
 Therefore, the number of vertices of~$\mCVD_k(S)$ for $1\leq k\leq m-1$ is at most~$2k^2$.
\end{lemma}
\begin{proof}
Recall that $\overline{U}_j = \sum_{i=0}^j (\bar{u}_{2,i}(S) + (j-i+1)\bar{u}_{1,i}(S))$
and that $\bar{u}_{c,j}(S)$ counts the number of $c$-chromatic unbounded edges in $\mCVD_{j+1}(S)$.
Each unbounded edge in $\mCVD_{j+1}$ corresponds to a quadrant $Q$.
Without loss of generality, we only consider those quadrants $Q$ whose bounding rays
are to the right and downwards, respectively.
Then, the following properties hold:
\begin{enumerate}[(i)] 
 \item the horizontal ray bounding $Q$ should contain the top-most point in~$S_i$
  for some color~$i\in K$ with~$S_i \subset Q$,
 \item the vertical ray bounding $Q$ should contain the left-most point in~$S_i$
  for some color~$i'\in K$ with~$S_{i'} \subset Q$,
 \item the exterior $\overline{Q} = \Plane\setminus Q$ of $Q$
 intersects exactly $j$ colors from $K\setminus \{i, i'\}$.
\end{enumerate}

To bound the number of those quadrants satisfying the above properties,
we consider the grid $\overline{G}$ obtained by drawing
a horizontal line through the top-most point from each color
and a vertical line through the left-most point from each color.
Let $\overline{G}(a,b)$ for $0\leq a,b\leq m-1$ be the grid point
that is the $(a+1)$-st from the left and the $(b+1)$-st from above.
Regard each of the $2m$~lines of $\overline{G}$ is given the same color as its original point.
Let $H_a \subseteq K$ be the set of $a$ colors of horizontal lines above the $(a+1)$-st horizontal line,
and $V_b \subseteq K$ be the set of $b$ colors of vertical lines to the left of the $(b+1)$-st vertical line.
Each grid point of $\overline{G}$ is called
\emph{monochromatic} if it is the intersection of the lines of a common color,
or \emph{bichromatic}, otherwise.
By construction,
each row or column of $\overline{G}$ has exactly one monochromatic point
and the others are bichromatic.
For each $a$, let $b_a$ be such that $\overline{G}(a, b_a)$ is monochromatic.

Fix $0\leq j\leq m-2$.
Define $w(a,b)$ to be the \emph{contribution} of $\overline{G}(a,b)$ to $\overline{U}_j$.
More precisely, we have $w(a,b) = 1$
if $\overline{G}(a,b)$ is the apex of
a quadrant $Q$ corresponding to a new $2$-chromatic unbounded edge in $\mCVD^*_{i+1}(S)$ for $i \leq j$,
$w(a,b) = j-i+1$
if that corresponds to a new $1$-chromatic unbounded edge in $\mCVD^*_{i+1}(S)$ for $i \leq j$,
or $w(a,b) = 0$, otherwise.

In the following, we want to find an upper bound on $w(a) = \sum_b w(a,b)$ for each $0\leq a\leq m-1$.
It is obvious that $w(a,b) = 0$ if $a > j$ or $b>j$,
so $w(a) = 0$ for $a>j$ and $w(a) = \sum_{b=0}^j w(a,b)$.
Also, we have $w(a,b) = 0$ for $b > b_a$ by property (ii).
Now, fix $a$ with $0\leq a\leq j$, and consider $b$ from $0$ to $j$ one by one.
There are two cases: either $b_a \geq j+1$ or $b_a < j+1$.

Suppose the former case where $b_a \geq j+1$.
Then, for $b$ with $b \leq j$,
$\overline{G}(a,b)$ is bichromatic and thus
we have $w(a,b) = 0$ if either $|H_a \cup V_b| > j$ or
the color of the $(b+1)$-st vertical line belongs to $H_a$;
or $w(a,b) \leq 1$, otherwise.
So, as $b$ increases,
the cardinality of $H_a \cup V_b$ increases by one
only when we encounter such $b$ that $w(a,b)=1$.
This implies that $w(a) = \sum_b w(a,b) \leq j-a+1$.

Next, consider the case of $b_a < j+1$.
Let $x_a := |H_a \cup V_{b_a}|$.
If $x_a > j$, then the above argument still holds to see that $w(a) \leq j-a+1$.
So, assume that $x_a \leq j$.
Then, we have $w(a, b_a) \leq j-x_a+1$ and $w(a,b) = 0$ for $b>b_a$.
On the other hand, for $0\leq b < b_a$,
we have $w(a,b) = 0$ if the color of the $(b+1)$-st vertical line is a member of $H_a$;
or $w(a,b) \leq 1$, otherwise.
A similar argument as above shows that
$w(a,0) + \cdots + w(a, b_a-1) \leq x_a - a$.
Hence, we conclude in this case that
\[ w(a) \leq (x_a - a) + (j - x_a +1) = j-a+1.\]

This implies that $\sum_{0\leq a \leq j} w(a) \leq \frac{1}{2}(j+1)(j+2)$,
and thus the claimed bound $\overline{U}_j \leq 2(j+1)(j+2)$ follows
because there are three more different directions of quadrants,
which can be handled analogously.
Combining this with the equality in Theorem~\ref{thm:kcvd_general},
we obtain
\[ \overline{V}_j = \overline{U}_j - (j+1)(j+2) \leq (j+1)(j+2),\]
for any $0\leq j \leq m-2$.
Hence, by Lemma~\ref{lem:CVD_v},
the number of vertices of~$\mCVD(S)$ is
\begin{align*}
 & \bar{v}_{3,k-1}(S) + \bar{v}_{3,k-2}(S) + \bar{v}_{2,k-1}(S)  \\
 & \leq
  (k-1)k + k(k+1) - 2 \sum_{i=0}^{k-2} \bar{v}_{2,i}(S) - \sum_{i=0}^{k-1}(2k-2i-1)\bar{v}_{1,i}(S)
  \leq 2k^2,
\end{align*}
as claimed.
\end{proof}

Summarizing, we obtain:
\begin{theorem} \label{thm:kcvd_L1}
 Let $S$ be a set of $n$ colored points with $m$ colors in the $L_\infty$ or $L_1$ plane.
 For $1\leq k\leq m-1$, the number of vertices in $\CVD_k(S)$ is at most
 $ \min\{4k(n-k)-2n, 4(n-k)^2\}$
 and the number of vertices in $\mCVD_k(S)$ is at most
 $\min\{4k(n-k)-2n, 2k^2\}$.
\end{theorem}

The above approach also works for polygonal convex distances,
concluding the following.
\begin{corollary} \label{coro:kcvd_convex_polygonal} 
 Let $B$ be a convex $2b$-gon with $2b\geq 4$, centrally symmetric around the origin,
 and
 $S$ a set of $n$ colored points with $m$ colors.
 For $1\leq k\leq m-1$,
 $\CVD_k(S)$ and $\mCVD_k(S)$ under the convex distance function based on $B$
 consist of at most $\min\{4k(n-k)-2n, 2(b^2-b)(n-k)^2\}$
 and $\min\{4k(n-k)-2n, 2(b^2-b-1)k^2\}$ vertices, respectively.
\end{corollary}
\begin{proof}
The same bound~$4k(n-k)-2n$ of Corollary~\ref{coro:kcvd_convex_distance}
holds for this case.

Let $D$ be the set of $b$~orientations of the sides of~$B$.
Take any pair of two orientations $\theta_1, \theta_2\in D$,
and consider the quadrilateral~$B'$ formed by stretching the four sides of~$B$
whose orientations are either $\theta_1$ or $\theta_2$.
We build the grid~$G$ by drawing two lines parallel to $\theta_1$ and $\theta_2$
through each point in~$S$.
Then, the same argument as in the proof of Lemma~\ref{lem:kcvd_L1_grid}
concludes that the number of new vertices in~$\CVD_k(S)$ such that
their corresponding copies of~$B'$ have its corner lies on the grid points of~$G$
is at most $2(n-k-1)^2$.
Since there are $\binom{b}{2}$ such pairs of orientations,
the number of vertices of~$\CVD_k(S)$ is bounded by
\[ \binom{b}{2}(2(n-k-1)^2 + 2(n-k)^2) \leq 2(b^2 - b)(n-k)^2.\]

Similarly, we can show that
\[ \overline{U}_j \leq 2\binom{b}{2}(j+1)(j+2)\]
for any~$0\leq j\leq m-2$,
by considering each pair of orientations in~$D$ and
applying the same argument as in the proof of Lemma~\ref{lem:kcvd_L1_barU}.
Then, Theorem~\ref{thm:kcvd_general} implies that
\[ \overline{V}_j = \overline{U}_j - (j+1)(j+2) \leq (b^2-b-1)(j+1)(j+2),\]
and thus the claimed upper bound $2(b^2-b-1)k^2$
on the number of vertices of~$\mCVD_k(S)$ for $1\leq k\leq m-1$.
\end{proof}

Remark that 
a more careful analysis could reduce the factor depending on~$b$, and
relax the central symmetry of~$B$\changed{.}

\section{Iterative algorithms for color Voronoi diagrams}\label{sec:alg}

In this section, we present an iterative approach to compute the
order-$k$ color Voronoi diagrams
and 
their refined counterparts for an increasing order of $k$.
Recall that $S$ is a set of $n$~sites
associated with distance functions~$\delta_s$ for~$s\in S$.
\changed{We assume the general position assumption on~$\{\delta_s\}_{s\in S}$ 
given in Section~\ref{sec:cvda}.}
We first establish some key structural properties, which add the
concept of color
to well-known properties of order-$k$ Voronoi diagrams.

Consider a face $f$ of an order-$k$
Voronoi region $\VR_k(H; S)$ of $\CVD_k(S)$, or a region $\mVR_k(H; S)$ of
$\mCVD_k(S)$,  
where $H\subseteq K$ with $|H|=k$.
Recall that $S_H \subseteq S$ is
the set of sites whose colors are included in~$H$.
Let $S_f\subseteq S$ \changed{be}
the set of sites that, 
together with sites in $S_H$, define the edges
along the boundary of $f$.
The following properties are derived directly from the definitions. 

\begin{lemma}
  \label{obs:no-color-in-H}
  No site in $S_f$ has a color that is included in $H$.
\end{lemma}
\begin{proof}
  Consider an edge $e$ along
  the boundary of $f$. Assuming that $f$ is a face of  $\VR_k(H; S)$, let $\VR_k(H';
  S)$ be the region incident to $e$ on the other side of $f$.
  A point $x$ on $e$ is equidistant from a site $s_h\in S_H$ of color
  $c_h$ and a site $s_f\in S_f$ of color $c_f$, where $c_h\neq c_f$.
  But if $c_f\in H$,  then $x$ would lie in
  $\VR_k(H; S)$  as $H$ would still be the set of
  the $k$ nearest colors to $x$, deriving a contradiction; thus,
  $c_f\not\in H$ and $S_f\cap S_H=\emptyset$.
  The proof is analogous for a face  of $\mCVD_k(S)$.
\end{proof}

\begin{lemma}
  \label{lem:CVD-subdivisions} 
  Let $f\subseteq \VR_k(H;S)$ be a face of $\CVD_k(S)$ for $1\leq k\leq m-1$.
  It holds that:
  
  \begin{enumerate}[(i)]
  \item $\CVD_1(S_f) \cap f =  \CVD_{k+1}(S) \cap f$
    and $\VD(S_f) \cap f =  \CVD^*_{k+1}(S) \cap f$. 
  \item $\FCVD(S_H) \cap f = \CVD_{k-1}(S) \cap f$ and
    $\FCVD^*(S_H) \cap f = \CVD^*_{k}(S) \cap f$. 
  \end{enumerate}
\end{lemma} 
  
\begin{proof}
  By the definition of order-$k$ diagrams, it is clear that
  $\CVD^*_{k+1}(S) \cap f= \VD(S\setminus S_H)\cap f$; and  by
  Lemma~\ref{obs:no-color-in-H}, $S_f\subseteq S\setminus S_H$.
  In~$\VD(S\setminus S_H)$, no region of this diagram can be entirely 
  enclosed in~$f$, as no site of $S\setminus S_H$ can lie in~$f$;
  further, only regions of sites in~$S_f$ can
  intersect the boundary of~$f$. 
  Thus, $\VD(S_f) \cap f=\VD(S\setminus S_H) \cap f$ 
  and claim~(i) follows.

  By the definition of order-$k$ diagrams, 
  and the fact $f\subseteq \VR_k(H; S)$, the following hold:
  $\CVD_{k-1}(S)\cap f=\CVD_{k-1}(S_H)\cap f$ and $\CVD_{k}^*(S)\cap f=\CVD_{k}^*(S_H)\cap f$.
  Since $\FCVD(S_H) = \CVD_{k-1}(S_H)$ and $\FCVD^*(S_H) = \CVD^*_{k}(S_H)$,  
  claim~(ii) follows.
\end{proof} 

\changed{
 We use Lemma~\ref{lem:CVD-subdivisions}(i) to iteratively compute $\CVD_{k+1}^*(S)$, given $\CVD_{k}^*(S)$. 
 Lemma~\ref{lem:CVD-subdivisions}(ii) indicates that
 superimposing $\CVD_{k}(S)$ and $\CVD_{k-1}(S)$ results in $\CVD^*_{k}(S)$ 
 with its $1$-chromatic edges removed.
}

Analogous claims hold for the maximal diagrams, 
however, for an unbounded face~$f$ of $\mCVD_{k}(S)$,
the set~$S_f$ is no longer adequate to derive the
portion of $\mCVD_{k+1}(S)$ that lies within~$f$.
We need the set $\Sfplus\subseteq S\setminus S_H$, 
which defines the unbounded faces of $\mCVD_{k+1}(S)\cap f$.

\begin{figure}[ht]
\begin{center}
\includegraphics[width=.95\textwidth]{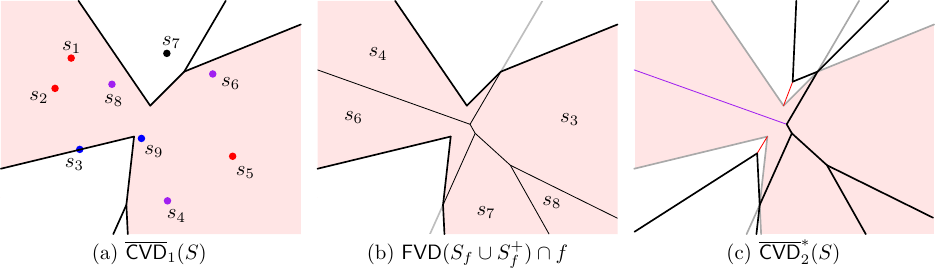}
\end{center}
\caption{
 Illustration to Lemma~\ref{lem:mCVD-subdivisions}
 for unbounded faces of $\mCVD_k(S)$
 with the same set~$S=\{s_1,\ldots, s_9\}$ of colored points as in
 \figurename s~\ref{fig:cvd} and~\ref{fig:mcvd} under the Euclidean metric.
 The shaded regions in (a)--(c) depict a face~$f\subseteq \VR_1(H; S)$ of~$\mCVD_1(S)$ 
 where $H$ consists of a single color (for the red points).
 In this case, $S_f = \{s_3, s_4, s_6, s_7\}$,
 while an additional site~$s_8$ defines an unbounded face in~$\mCVD^*_2(S) \cap f$.
 So, $\Sfplus \setminus S_f = \{s_8\}$ and 
 $\mCVD^*_2(S) \cap f = \FVD(S_f \cup \Sfplus) \cap f$.
 }
\label{fig:Sfplus}
\end{figure}

\begin{lemma}
  \label{lem:mCVD-subdivisions} 
  Let $f\subseteq \mVR_k(H;S)$ be a face of $\mCVD_k(S)$ for $1\leq
  k\leq m-1$.
  Let $\Sfplus\subseteq S\setminus S_H$ be the set of sites
  that define unbounded faces in  $\mCVD_{k+1}(S)\cap f$\changed{; if
    $f$ is bounded, $\Sfplus=\emptyset$.}
  The following hold:
  
  \begin{enumerate}[(i)] 
  \item
    $\mCVD_1(S_f\cup \Sfplus) \cap f =  \mCVD_{k+1}(S) \cap f$
    and $\FVD(S_f \cup \Sfplus) \cap f =  \mCVD^*_{k+1}(S) \cap f$.
  \item $\HVD(S_H) \cap f = \mCVD_{k-1}(S) \cap f$ and
    $\HVD^*(S_H) \cap f = \mCVD^*_{k}(S) \cap f$.
  \end{enumerate}
\end{lemma}
\begin{proof}
  Analogously to Lemma~\ref{lem:CVD-subdivisions}, $S_f\subseteq S\setminus S_H$ and 
  $\mCVD^*_{k+1}(S) \cap f= \FVD(S\setminus S_H)\cap f$, where
  $f\subseteq \mVR_k(H;S)$ is a face of $\mCVD_k(S)$.
  \changed{
  Further, $\FVD(S\setminus S_H)$ has only unbounded regions, 
  thus a face of $\FVD(S\setminus S_H)$ may be
  enclosed in $f$ only if it is unbounded in the same
  directions as~$f$.}
  In addition, only sites in~$S_f$ can have a region \added{in $\FVD(S\setminus S_H)$} that
  intersects the boundary of~$f$.
  The unbounded
  edges of $\mCVD^*_{k+1}(S) \cap f$ are clearly all new, by the
  definition of an unbounded face~$f$.
  Thus, $\FVD(S\setminus S_H)\cap f=\FVD(S_f \cup \Sfplus) \cap f$, 
  where $\Sfplus=\emptyset$ if $f$ is
  bounded; hence claim~(i) follows.
  See \figurename~\ref{fig:Sfplus}.

  Claim~(ii) is analogous to Lemma~\ref{lem:CVD-subdivisions}(ii).
  Since $f\subseteq \mVR_k(H; S)$, 
  it holds that $\mCVD_{k-1}(S)\cap f=\mCVD_{k-1}(S_H)\cap f$ 
  and $\mCVD_{k}^*(S)\cap f=\mCVD_{k}^*(S_H)\cap f$.
  Since $\HVD(S_H) = \mCVD_{k-1}(S_H)$ 
  and $\HVD^*(S_H) = \mCVD_{k}^*(S_H)$, 
  the claim follows.
\end{proof}

To iteratively compute $\mCVD_{k+1}(S)$, \changed{given  $\mCVD_{k}(S)$,} we first need
to identify \changed{the}
sites that define the new unbounded edges of
$\mCVD_{k+1}(S)$.
This information, however,  is not encoded in $\mCVD_{k}(S)$, unlike \added{the
  minimal diagrams}.
We give a condition, related to condition~\textbf{V3}, 
under which we can use $\CVD^*_{k+1}(S)$ to derive the information missing
from  $\mCVD^*_{k+1}(S)$.
This condition is satisfied
if, for example, $S$ is a set of points and $\delta_s$ for $s\in S$
is \changed{a}
convex distance based on a smooth body.

\begin{itemize}
\item[\textbf{V3$^\prime$}]
The unbounded faces of $\VD(S')$ and $\FVD(S')$ are defined by the
same sequence  of sites, for any~$S'\subseteq S$.
\end{itemize}

\begin{lemma}
	\label{lem:V3-prime} 
	Condition \textbf{V3$^\prime$} implies that the unbounded faces of
	$\CVD^*_{k}(S)$ and of $\mCVD^*_{k}(S)$ are defined by the same
	sequence of sites for any $1\leq k\leq m$. (See \figurename~\ref{fig:condition_V3prime}.)
\end{lemma}

\begin{proof}
  We perform induction on the order~$k$.
  The base case for~$k=1$ holds by condition~\textbf{V3$^\prime$}.
 Suppose that the claim holds for a given~$k$ with~$1\leq k< m$. 
 Then for any two consecutive unbounded edges of $\CVD_{k}(S)$ 
 that delimit a face $f\subseteq \VR_k(H;S)$, $|H|=k$, 
 there is a corresponding face $f'\subseteq \mVR_k(H';S)$, $|H'|=k$, 
 in $\mCVD_{k}(S)$, delimited by unbounded
 edges defined by the same pairs of sites.
 To prove the lemma,
 we strengthen the induction hypothesis by further assuming that $H=H'$. 
 The extended hypothesis clearly holds for~$k=1$.

\begin{figure}[h]
\begin{center}
\includegraphics[width=\textwidth]{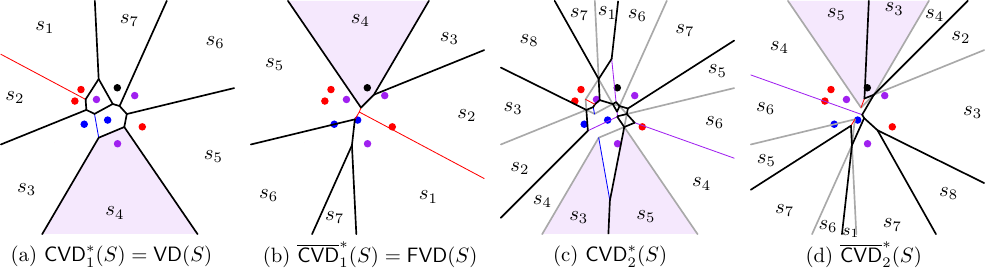}
\end{center}
\caption{
  Illustration \changed{of}
  condition~\textbf{V3$^\prime$} and Lemma~\ref{lem:V3-prime}
 with the same set~$S=\{s_1,\ldots, s_9\}$ as in~\figurename s~\ref{fig:cvd}--\ref{fig:mcvd} 
 under the Euclidean metric.
 In (a)(b), $\VD(S')$ and $\FVD(S')$
 have the same sequence of sites that define unbounded faces for any~$S'\subseteq S$,
 so condition~\textbf{V3$^\prime$} holds.
 The shaded region in~(a) is a face~$f$ of $\CVD_1(S)$
 corresponding to face~$f'$ of $\mCVD_1(S)$ shaded in~(b).
 In (c)(d), shaded regions show how the portions of $f$ and $f'$ at infinity
 are subdivided in $\CVD^*_2(S)$ and $\mCVD^*_2(S)$.
 }
\label{fig:condition_V3prime}
\end{figure}

By Lemma~\ref{lem:CVD-subdivisions}(i), we have
 $\CVD^*_{k+1}(S) \cap f=\VD(S\setminus S_H) \cap f=\VD(S_f) \cap f$.
 Analogously, by Lemma~\ref{lem:mCVD-subdivisions}(i),
 $\mCVD^*_{k+1}(S) \cap f'=\FVD(S\setminus S_H) \cap f'=\FVD(S_{f'}\cup\Sfplusm)\cap f'$.
 Since the unbounded edges delimiting $f$ and $f'$ are corresponding,
 the inductive hypothesis and condition~\textbf{V3$^\prime$} imply that 
 the sequence of sites  that define the
 unbounded faces of $\VD(S\setminus S_H)\cap f$ and 
 of $\FVD(S\setminus S_H) \cap f'$ coincide.
 Thus the claim follows for $\CVD^*_{k+1}(S) \cap f$ and
 $\mCVD^*_{k+1}(S) \cap f'$.
 Since our choice of~$f$ is arbitrary, 
 the claim holds for $\CVD^*_{k+1}(S)$ and $\mCVD^*_{k+1}(S)$.
 See \figurename~\ref{fig:condition_V3prime}.
 
 It remains to show that the extended hypothesis continues to hold. 
 Consider an unbounded face~$f_1$ of $\VD(S\setminus S_H)\cap f$ 
 that is incident to a new unbounded edge~$e$ of $\CVD^*_{k+1}(S) \cap f$. 
 Then $f_1\subseteq \VR_{k+1}(H_1;S)$, 
 where $H_1$ consists of those colors in~$H$ and 
 the color of the site that defines~$e$ from the side of~$f_1$.
 Face~$f_1$ corresponds to a face~$f_1'$
 of $\FVD(S\setminus S_H)\cap f'$, 
 which is incident to the unbounded edge~$e'$ corresponding to~$e$.
 Thus, $f_2\subseteq \mVR_{k+1}(H_2;S)$, 
 where $H_2$ consists of those colors in~$H$
 and the color of the site that defines~$e'$ on the side of~$f_2$,
 which is the same as the site that defines~$e$. 
 Thus, we have $H_2=H_1$ and
 the extended hypothesis holds for order~$k+1$ as well.
\end{proof}

Now, we assume that \changed{the sites~$S$ and their distance functions~$\delta_s$}
fall under the model
of \emph{abstract Voronoi diagrams}~\cite{k-cavd-89}.
Specifically, together with the general position assumption described
in Section~\ref{sec:cvda},
we also assume the following, for every subset $S' \subseteq S$:
\begin{itemize}
\item The regions of~$\VD(S')$ are nonempty and connected.
\item The bisector of any two sites is an unbounded simple curve homeomorphic to a
  line.
\end{itemize}
Furthermore, we assume that \changed{any distance or bisector}
can be computed in $O(1)$ time.
We then conclude the following.

\begin{theorem} \label{thm:alg} 
 Let $S$ and $\{\delta_s\}_{s\in S}$ 
 be a
 set of $n$~colored sites with $m$~colors and distance functions
 that satisfy the conditions of abstract Voronoi diagrams.
 Then, for~$1\leq k \leq m$,
 in $O(k^2n + n\log n)$ expected time or in $O(k^2n \log n)$ worst-case time,
 we can compute
 $\CVD^*_1(S), \ldots, \CVD^*_k(S)$.
 If in addition condition~\textbf{V3$^\prime$} holds, then 
 we can also compute
 $\mCVD^*_1(S), \ldots, \mCVD^*_k(S)$
 in the same time bound.
 If $S$ consists of points and $\delta_s(x)$ is the Euclidean distance to~$s\in S$, 
 the time bound is reduced to $O(k^2n + n\log n)$ in the worst case.
\end{theorem}
\begin{proof}
  Let $i$ be an integer with $1\leq i\leq k-1$.
  Consider a face~$f$ of $\CVD_i(S)$ that belongs to an order-$i$
  Voronoi region $\VR_i(H; S)$, for a set~$H$ of $i$~colors.
  Recall the sets~$S_H$ and~$S_f$ as defined above, where $S_f\cap S_H = \emptyset$
  \added{by Lemma~\ref{obs:no-color-in-H}.}

  By Lemma~\ref{lem:CVD-subdivisions}, the Voronoi diagram $\VD(S_f)$,
  truncated within $f$, reveals exactly the order-$(i+1)$ subdivision
  within the face~$f$, $\CVD^*_{i+1}(S)\cap f$.
  Since $\VD(S_f)$ is an instance of abstract Voronoi diagrams, we can
  compute  $\CVD^*_{i+1}(S)\cap f$ in expected $O(|S_f|)$~time by
  the randomized incremental technique of~\cite{P23},
  or in worst-case $O(|S_f|\log |S_f|)$ time by standard means, see e.g., \cite{alk-vddt-13}. 
  If $S$ consists of points in~$\Plane$ and $\delta_s(x) = \|x-s\|_2$ 
  is the Euclidean distance to each $s\in S$,
  then
  $\CVD^*_{i+1}(S)\cap f$ can be computed in $O(|S_f|)$ worst-case time~\cite{agss-ltacvdcp-89}.

  Then $\CVD_{i+1}(S)\cap f$ can be derived in two steps.
  First delete any $1$-chromatic edges of $\VD(S_f)$ and 
  unify the faces incident to the deleted edges. 
  This yields the overlay of $\CVD_{i}(S)$ and $\CVD_{i+1}(S)\cap f$, 
  which is  $\CVD^*_{i+1}(S)\cap f$ with its $1$-chromatic edges removed.
  Then, remove the edges along the  boundary of~$f$, 
  while unifying their incident faces, 
  which belong to the same set of $i+1$~colors.
  Note that for each edge~$e$ removed, the two incident faces get unified 
  into a new face~$f'(e)$ of $\CVD_{i+1}(S)$, 
  which belongs to a set~$H'$ of $i+1$~colors;
  the removed edge~$e$ belongs to $\FCVD(S_{H'})\cap f'(e)$.

  To obtain the entire $\CVD^*_{i+1}(S)$, we repeat the
  process for every face~$f$ of~$\CVD_{i}(S)$.
  \changed{
  As discussed in Corollary~\ref{coro:kcvd_AVD},
  conditions~\textbf{V1} and~\textbf{V2} hold, thus,  $\sum_{f} |S_f|
  = O(in)$.
  Hence, for computing~$\CVD^*_{i+1}(S)$ given $\CVD^*_i(S)$,
  we spend $O(in)$ expected or $O(in \log n)$ worst-case time,
}
 plus time proportional to the combinatorial complexity of $\CVD^*_{i+1}(S)$,
  which is $O((i+1)n)$ by Theorem~\ref{thm:kcvd_general} and Corollary~\ref{coro:kcvd_AVD}.
  Therefore, the total time complexity of our algorithm is bounded as claimed.

  Assuming condition~\textbf{V3}$^\prime$, 
  we can compute $\mCVD^*_{i+1}(S)$ and $\mCVD_{i+1}(S)$ analogously,
  however, given both $\mCVD_{i}(S)$ and $\CVD_{i}(S)$.
  We first compute $\CVD^*_{i+1}(S)$ from $\CVD_{i}(S)$ 
  in order to extract the sequence of sites 
  that define the unbounded faces of $\CVD^*_{i+1}(S)$, 
  which by Lemma~\ref{lem:V3-prime} coincides with 
  the sequence of sites that define the unbounded faces of  $\mCVD^*_{i+1}(S)$.
  In particular, for each pair of unbounded edges 
  that delimit an unbounded face~$f$ in $\mCVD_{i}(S)$,
  we identify the corresponding pair of unbounded edges in $\CVD_{i}(S)$; 
  then we traverse the unbounded edges of $\CVD^*_{i+1}(S)$ 
  that lie between them, and transform them to the sequence $\Jfplus$ of sites 
  that define the unbounded edges of $\mCVD_{i}(S)\cap f$.
  Once $\Jfplus$ and $\Sfplus$ are identified, 
  for every unbounded face of $\mCVD_{i}(S)$, 
  the computation of $\mCVD^*_{i+1}(S)$ and $\mCVD_{i+1}(S)$ is analogous: 
  for each face~$f$, we compute the farthest-site Voronoi diagram 
  $\FVD(S_f\cup S_{f}^+)\cap f$, 
  where $S_{f}^+ =\emptyset$ if $f$ is bounded.
  Given the boundary of~$f$ and the ordering of~$\Jfplus$, 
  this can be done in expected linear time by applying
  the randomized incremental construction of~\cite{JP23,P23}.
  (Note that we can easily compute $\FVD(S_f\cup \Sfplus)\cap \partial f$,  
  from $\Jfplus$, in linear time, 
  where  $\partial f$ denotes the boundary of~$f$, 
  as obtained by superimposing~$f$ and a large-enough bounding circle; 
  we can then apply~\cite{JP23,P23}).
  For points in the Euclidean metric,  
  $\FVD(S_f\cup S_{f}^+)\cap f$ can be computed in deterministic
  linear time~\cite{agss-ltacvdcp-89}, 
  given the sequence of sites that appear along the boundary of~$f$.
  Alternatively, we can also compute $\FVD(S_f\cup S_{f}^+)$ 
  in $O(|S_f \cup S_{f}^+| \log |S_f \cup S_{f}^+|)$ time 
  by standard techniques~\cite{alk-vddt-13}.  
  Since condition~\textbf{V3}$^\prime$ implies condition~\textbf{V3},
  the complexity of~$\mCVD^*_i(S)$ is also bounded by $O(i(n-i))$
  by Theorem~\ref{thm:kcvd_general_2n-4} and Lemma~\ref{lem:CVD_v}.
  The claimed time bounds are thus derived.
\end{proof}

The convex distance functions satisfy the conditions of Theorem~\ref{thm:alg},
hence we have:
\begin{corollary} \label{cor:alg} 
Let $B$ be a convex and compact body in~$\Plane$ of a constant complexity
that contains the origin in its interior.
Given a set~$S$ of $n$~colored points in~$\Plane$ with $m$~colors and
an integer $1\leq k\leq m$,
we can compute
$\CVD^*_1(S), \ldots, \CVD^*_k(S)$
in $O(k^2n + n\log n)$ expected time or in $O(k^2n \log n)$ worst-case time.
If $B$ is smooth, then we can also compute 
$\mCVD^*_1(S), \ldots, \mCVD^*_k(S)$ in the same time bound.
\end{corollary}
\begin{proof}
As discussed in Section~\ref{sec:kcvd_general} this case
falls under the \changed{umbrella} of abstract Voronoi diagrams; 
furthermore, if $B$ is smooth,
condition~\textbf{V3$^\prime$} holds, 
as shown in Lemma~\ref{lem:V3-convex_distance} and its proof.
Therefore, Theorem~\ref{thm:alg} applies and the corollary follows.
\end{proof}

\begin{corollary}\label{coro:kcvd_not_smooth}
Let $B$ be a convex $2b$-gon, centrally symmetric around the origin,
where $b\geq 2$ is a constant.
Given a set~$S$ of $n$~colored points with $m$~colors
in~$\Plane$,  and an integer $1\leq k\leq m$, we can
compute $\CVD^*_1(S), \ldots, \CVD^*_k(S)$ in 
$O(k^2 (n-k) + n\log n)$
expected or 
$O(k^2 (n-k) \log n + n \log n)$
worst-case time.
We can then compute $\mCVD^*_1(S), \ldots, \mCVD^*_k(S)$ in additional
$O(k^3 + n)$ worst-case time.
\end{corollary}

\begin{proof}
  By Corollary~\ref{coro:kcvd_convex_polygonal}, the complexity of
  $\CVD_i$ under this metric is $ O(\min\{i(n-i),
  (n-i)^2\})$, where $1\leq i\leq k$. Following the proof of
  Theorem~\ref{thm:alg} and summing up over $ 1\leq i \leq k$, the 
  the claimed bounds regarding the minimal diagrams can be derived.
  However, some care is required with respect to the general position
  assumption and the specifics of the $2b$-gon metric.
  
  Let $\mathcal{N}(B)$ be the set of directions normal
  to the sides of $B$ pointing outwards; and let $\mathcal{D}(B)$ be the
  set of directions along the diagonals of $B$, see
  Figure~\ref{fig:directions}.
  The $2b$-gon ($b=2$) bisector of two points in general position is illustrated in
  Figure~\ref{fig:subdivision}.
  If the points are collinear along  a line normal to a direction in $\mathcal{N}(B)$, then their bisector contains 
  2-dimensional regions, which are not allowed by the general position
  assumption, see 
  Figure~\ref{fig:subdivision}(b).
  Following standard conventions, we can transform any such bisector 
  to a simple curve, by assigning an equidistant area to only one of
  the sites and keeping only the boundary curve as the bisector, see
  Figure~\ref{fig:subdivision}(c).
  Following this convention consistently, e.g., always choosing the
  ``clockwise most'' boundary of the equidistant area to be part of the bisector,
  the bisector system
  complies with the assumptions and the algorithms in the proof
  of Theorem~\ref{thm:alg} can be used.
  We adopt this convention in the rest of this proof.

  \begin{figure}[ht]
    \centering
    \includegraphics[height=3cm]{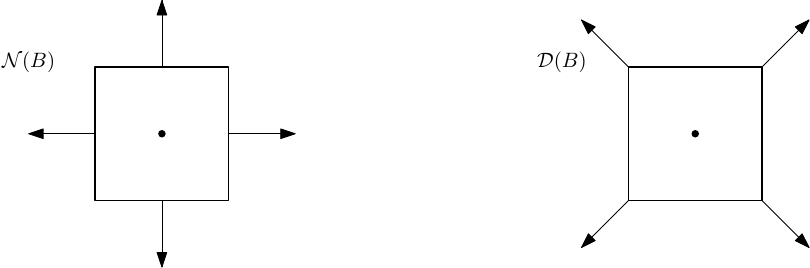}
    \caption{The directions in $\mathcal{N}(B)$ and $\mathcal{D}(B)$ for the $L_\infty$ metric ($b = 2$).}
    \label{fig:directions}
  \end{figure}

  Consider the diagrams $\VD(S)$ and  $\FVD(S)$ under the $2b$-gon
  metric; they both have unbounded faces in each
  of the directions of 
  $\mathcal{N}(B)$ defined by the minimum enclosing $2b$-gon of $S$.
  Furthermore, assuming that we consistently follow the same
  tie-breaking convention,
  the point
  that defines the face of $\VD(S)$ unbounded in direction $\vec{v}\in
  \mathcal{N}(B)$, and the point that defines the face of $\FVD(S)$ unbounded in direction
   $-\vec{v}$, coincide.
  The $\VD(S)$ can have multiple faces unbounded in the directions of
  $\mathcal{D}(B)$, whereas the complexity of $\FVD(S)$ is constant $O(b)$.
  We can use the directions of $\mathcal{N}(B)$ to further refine the faces of $\VD(S)$ and
  $\FVD(S)$,  see Figure~\ref{fig:subdivision}, and therefore also the faces of  $\CVD^*_k(S)$ and
  $\mCVD^*_k(S)$.
  The following property holds.

\begin{figure}[ht]
  \centering
  \includegraphics[height=3.5cm]{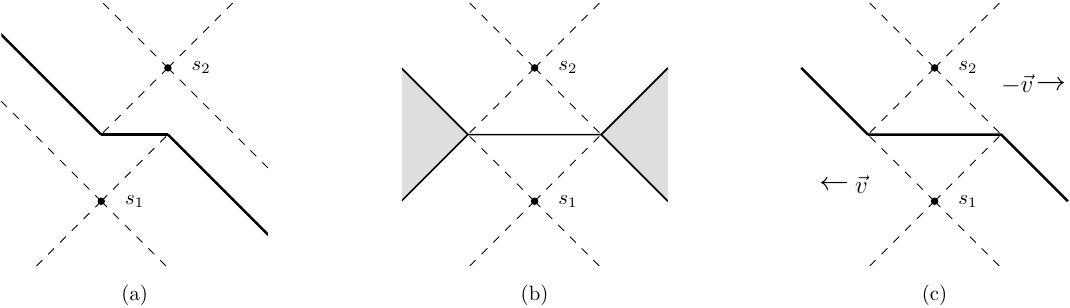}
  \caption{The $L_\infty$ Voronoi diagram of two points; their bisector
    is indicated in  solid lines;  (a) points in general position;
    (b) vertically collinear points, the shaded areas are equidistant from both points;
    (c) vertically collinear points under the adopted convention.
    }
  \label{fig:subdivision}
\end{figure}

\begin{lemma} \label{lem:extremal-points-Linf}
  The face of $\CVD^*_k(S)$ unbounded in direction
  $\vec{v}\in \mathcal{N}(B)$ and the face of $\mCVD^*_k (S)$
  unbounded in  direction $-\vec{v}$
  are associated with the same site $p\in S$ assuming the same tie-breaking conventions in both diagrams.
\end{lemma}

\begin{proof}
  Let $f\subseteq\VR_k(H;S)$ be the face of $\CVD^*_k(S)$ unbounded in 
  direction $\vec{v} \in \mathcal{N}(B)$ and let $p\in S_c$, where
  $c\in H$, be the
  point associated with $f$.
  That is, color $c \in H$ is the $k$-th
  nearest color from all points in  $f$,
  and $p$ is the nearest point of $S_c$ to all points in $f$,
  since $\CVD^*_{k}(S) \cap f=\FCVD^*(S_H) \cap f$, by Lemma~\ref{lem:CVD-subdivisions}(ii),
  and $f$ is already a fine face of $\CVD^*_{k}(S)$.
  The line $\ell$ through $p$ orthogonal to $\vec{v}$
  defines two open half-planes $h^+$ and $h^-$, unbounded in directions
  $\vec{v}$ and $-\vec{v}$, respectively, such that $h^+$
  contains at least one point of
  each color in $H\setminus\{c\}$  and no point of any other color,
  and the closure of $h^-$ entirely contains
  $S_c$ and $S\setminus S_H$.
  In case multiple points in $(S\setminus S_H)\cup S_c$ are collinear along
  $\ell$, the adopted tie-breaking convention indicates that $p$
  is the bottommost such point along $\ell$, where $\ell$ is oriented so
  that $h^+$ lies to its left.

  Consider the face $f'$ of  $\mCVD^*_k(S)$,  $f'\subseteq\mVR_k(H';S)$, unbounded in 
  direction $-\vec{v}$, and let $q$ be the point associated with $f'$.
  The line through $q$ orthogonal to $-\vec{v}$ defines two 
  half-planes that have exactly the same properties as $h^-$ and
  $h^+$.
  Thus, $H'=H$, and $p,q$  both lie on  $\ell$. 
  By the adopted tie breaking convention, $q$ must be extreme
  along $\ell$,  similarly to $p$, therefore $p$ and $q$ must coincide.

  \begin{figure}[ht]
    \centering
    \includegraphics[height=4cm]{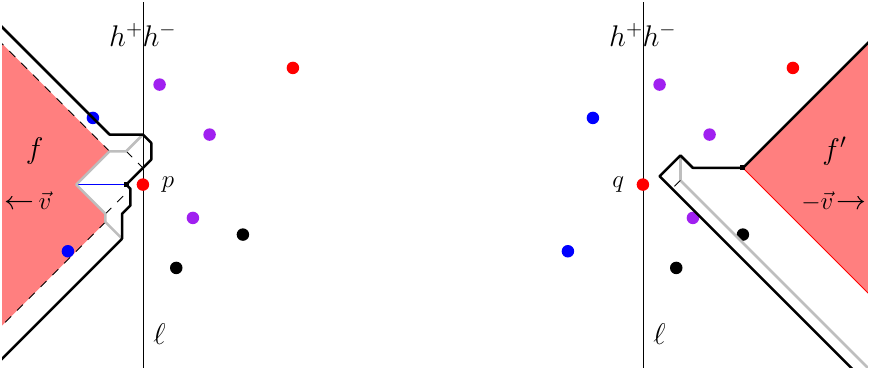}
    \caption{Illustration to Lemma~\ref{lem:extremal-points-Linf} for
      the $L_\infty$ metric. The faces $f\subseteq\VR_k(H;S)$ and 
      $f'\subseteq\mVR_k(H';S)$ are shown in red, region
      boundaries in solid black. In this case $k = 2$, the colors in $H$ are red and blue, the $k$-th color $c$ is red, and $p = q$.}
    \label{fig:extremal-points}
  \end{figure}  
\end{proof}

Let $f$ be a face of $\mCVD^*_{i}(S)$ unbounded in direction
$\vec{v}$.  By Lemma~\ref{lem:extremal-points-Linf}, we can compute
$\Sfplus$ by locating  in $\CVD^*_{i+1}(S)$ the face  $f'$ that is
unbounded in direction $-\vec{v}$, and assigning its associated
point as $\Sfplus$. Then $\FVD(S_f \cup \Sfplus) \cap f$ 
can be computed in $O(|S_f|)$ time, as the complexity of $\FVD(S_f\cup \Sfplus)$ is  constant.
Thus, we can derive $\mCVD_{i+1}(S)$, given $\mCVD_{i}(S)$ and $\CVD^*_{i+1}(S)$,  in time proportional to the
complexity of $\mCVD_{i}(S)$, which is $O(\min\{i(n-i), i^2\})$ by
Corollary~\ref{coro:kcvd_convex_polygonal}.
Summing up for $i=1$ to $k$,  we derive $O\left( \sum_{i=1}^k \min\{i(n-i),
i^2\} \right)= O(k^3)$
plus $O(n)$ time to compute $\mCVD^*_1(S)$.
This is in addition to the  time required  to compute
$\CVD^*_{i+1}(S)$, which has already been stated above.
\end{proof}

\section{More Applications of the Colorful Clarkson--Shor Framework} \label{sec:appl_CS}
The colorful Clarkson--Shor framework, as described in Section~\ref{sec:framework},
provides a general scheme to transform any set system that fits in
the original framework to its colored variant.
Once one has an upper bound on the number of (uncolored) configurations,
Theorems~\ref{thm:CS_general} and~\ref{thm:CS_general_uniform}
automatically imply general upper bounds on
the number of colored configurations of weight at most $k$.
In this section,
we demonstrate selected applications of the colorful Clarkson--Shor framework,
which result in new or, sometimes, known bounds 
on levels of arrangements of various objects of non-constant complexity.

\subsection{Envelopes of hyperplanes} \label{subsec:envelopes_hyperplanes}
We start with the arrangement of envelopes of hyperplanes in~$\Real^d$
for a constant $d\geq 2$.
Specifically,
let $S$ be a set of $n$~non-vertical hyperplanes in~$\Real^d$
and $\kappa \colon S \to K =\{1,\ldots, m\}$ be any color assignment.
For $i\in K$,
let $E_i$ be the lower envelope of \changed{the hyperplanes in~$S_i$}
and $\overline{E}_i$ \changed{be} their upper envelope.
We consider the arrangement $\arr = \arr(\{E_1, \ldots, E_m\})$ 
of $m$ lower envelopes
and the arrangement~$\overline{\arr} = \arr(\{\overline{E}_1, \ldots, \overline{E}_m\})$
of $m$ upper envelopes.
Our question is:
how many vertices are there in the arrangements~$\arr$ and $\overline{\arr}$ or 
in their levels?

We interpret this as an instance of the colorful Clarkson--Shor framework.
Let $\conf(S)$ be
the set of vertices of the arrangement $\arr(S)$ of \changed{the} $n$ hyperplanes in~$S$.
Let $\chi \subseteq S \times \conf(S)$
be a conflict relation such that $(s, v) \in \chi$ if and only if
$v\in \conf(S)$ lies \emph{above}~$s\in S$.
We also consider another relation $\bar{\chi} \subseteq S \times \conf(S)$
such that $(s, v) \in \chi$ if and only if
$v\in \conf(S)$ lies \emph{below}~$s\in S$.
This describes two symmetric (uncolored) CS-structures
$(S, \conf(S), \chi)$ and $(S, \conf(S), \bar{\chi})$.
Now, consider the colored configurations with respect to~$\kappa$
induced from $(S, \conf(S), \chi)$ and $(S, \conf(S), \bar{\chi})$,
denoted by $\conf(S,\kappa)$ and $\overline{\conf}(S,\kappa)$, respectively.
It then turns out that $\conf(S,\kappa)$
consists of the vertices of the arrangement~$\arr$ of $m$ lower envelopes,
while $\overline{\conf}(S,\kappa)$ consists of
the vertices of the arrangement~$\overline{\arr}$ of $m$ upper envelopes.
More precisely, for $1\leq c\leq d$ and $0\leq j\leq m-1$,
the set~$\conf_{c,j}(S,\kappa)$
consists of $c$-chromatic vertices of~$\arr$
below which there are exactly $j$~surfaces from~$\{E_i\}_{i\in K}$,
while $\overline{\conf}_{c,j}(S,\kappa)$ consists of
$c$-chromatic vertices of~$\overline{\arr}$
above which there are exactly $j$~surfaces from~$\{\overline{E}_i\}_{i\in K}$.

We then consider the standard point-to-hyperplane duality transformation~\cite{e-acg-87}
such that each point~$p = (a_1, a_2, \ldots, a_d) \in \Real^d$
is mapped to a non-vertical \changed{hyperplane~$\dual{p} \colon \{x_d = a_1 x_1 + \cdots + a_{d-1} x_{d-1} - a_d\}$},
and vice versa.
Letting $\dual{S}$ be the set of $n$~points in~$\Real^d$
that are dual to hyperplanes in~$S$,
a $c$-chromatic vertex of weight~$j$ of $\arr$ or of $\overline{\arr}$
(which belongs to $\conf_{c,j}(S)$ or $\overline{\conf}_{c,j}(S)$, respectively)
corresponds to a $c$-chromatic $j$-facet in~$\dual{S}$.
More precisely,
by the duality transformation,
there is a one-to-one correspondence between
$\conf_{c,j}(S)$ and the set of $c$-chromatic $j$-facets in~$\dual{S}$ that are \emph{upward},
(that is, those $j$-facets whose corresponding half-spaces are unbounded 
in the positive $x_d$-direction);
analogously,
there is a one-to-one correspondence between
$\overline{\conf}_{c,j}(S)$ and the of $c$-chromatic $j$-facets in~$\dual{S}$ that are \emph{downward}.
Hence, we have:
\begin{lemma} \label{lem:v-e_2}
For each $1\leq c\leq d$ and $0\leq j\leq m-c$, it holds that
\[ |\conf_{c,j}(S,\kappa)| + |\overline{\conf}_{c,j}(S)| = e_{c,j}(\dual{S}).\]
\end{lemma}

By Lemma~\ref{lem:v-e_2}, Corollary~\ref{coro:j-facets_general} implies:
\begin{corollary} \label{coro:levels_hyperplanes}
 The number of vertices in the $(\leq k)$-level of the arrangement of
 $m$ convex polyhedral hypersurfaces in~$\Real^d$ with a total of $n$ facets,
 each of which is the lower envelope of non-vertical hyperplanes, is
 $O(m^{\lfloor d/2 \rfloor -1} k^{\lceil d/2 \rceil} n^{\lfloor d/2 \rfloor})$
 in general
 and
 $O(k^{\lceil d/2 \rceil} n^{\lfloor d/2 \rfloor})$
 if the numbers of facets in
 each of the $m$ convex hypersurface is at most $\rho\cdot \frac{n}{m}$
 for a constant~$\rho\geq 1$.
\end{corollary}
Corollary~\ref{coro:levels_hyperplanes} is in fact
the dual version of Corollary~\ref{coro:j-facets_general}.
Note that for large $k$ with $k \geq \lfloor \frac{m}{d} \rfloor$,
the bounds in both corollaries becomes asymptotically the same as
the total number of $(\leq c)$-chromatic configurations,
$O(m^{d-1} n^{\lfloor d/2 \rfloor})$ and $O(n^{\lfloor d/2 \rfloor})$, respectively.
(see Theorems~\ref{thm:CS_general} and~\ref{thm:CS_general_uniform}).
Remark that the bounds in both corollaries for $d\leq 3$
match the original Clarkson--Shor bound
$O(k^{\lceil d/2 \rceil} n^{\lfloor d/2 \rfloor})$~\cite{cs-arscgII-89},
while the extra factor~$m^{\lfloor d/2 \rfloor -1}$ in higher dimensions~$d\geq 4$
is a bit disappointing.
Indeed, Aronov, Bern, and Eppstein have proved
that the total complexity of~$\arr$ is bounded by
$O(m^{\lceil d/2 \rceil} n^{\lfloor d/2 \rfloor})$,
but their unpublished manuscript~\cite{abe-apa-95} currently seems to be lost~\cite{a-pc-24}.
Katoh and Tokuyama~\cite{kt-klcs-02} have proved the bound of $O(k^{2/3}n^2)$ 
on the single $k$-level in~$\Real^3$.

The set~$S$ of hyperplanes in~$\Real^d$ is called in \emph{convex position}
if the set~$\dual{S}$ of dual points is in convex position.
By the duality and Lemma~\ref{lem:v-e_2},
when $S$ consists of planes in~$\Real^3$ in convex and general position,
Theorem~\ref{thm:j-facets_convex_3d} implies an exact upper bound
on the total number of vertices of the $k$-level of~$\arr$ from below
and of the $k$-level of~$\overline{\arr}$ from above.

\begin{corollary} \label{coro:levels_hyperplanes_convex}
 With the notations declared above for~$d=3$,
 suppose $S\subset \Real^3$ is in convex and general position.
 Then, for each~$1\leq k \leq m$,
 the total number of vertices in the $k$-level of~$\arr$ from below and
 in the $k$-level of~$\overline{\arr}$ from above is at most
 \[ \begin{cases}
      2n-4 			& k=1 \\
      6(k-1)(n-k)-4	& 2\leq k \leq m-1 \\
      4(m-1)(n-m+1)-2n	& k=m
    \end{cases}.
  \]
The exact numbers are achieved when $m=n$, that is, 
each $S_i$ consists of a single hyperplane.
\end{corollary}
\begin{proof}
By Lemma~\ref{lem:v-e_2}, the total number of vertices we are interested in is exactly
\[ e_{3,k-1}(\dual{S}) + e_{3,k-2}(\dual{S}) + e_{3,k-3}(\dual{S})
  + e_{2,k-1}(\dual{S}) + e_{2,k-2}(\dual{S}) + e_{1,k-1}(\dual{S})\]
from which Theorem~\ref{thm:j-facets_convex_3d} directly implies
the claimed exact bounds.
\end{proof}

Notice that the $k$-level of the two arrangements~$\arr$ and~$\overline{\arr}$
described in Corollary~\ref{coro:levels_hyperplanes_convex}
corresponds to the refined color Voronoi diagram~$\CVD^*_k(S')$ or $\mCVD^*_k(S')$
under the Euclidean metric of any set $S'$ of colored points
in~$\Plane$ \changed{such that  $\dual{S} = \lift{(S')}$ is 
the set of points in~$\Real^3$ lifted onto the unit parabola.}
(Recall the discussions above and in Section~\ref{sec:framework}.2.)
Thus, the total number of vertices in $\CVD^*_k(S')$ and $\mCVD^*_k(S')$
is bounded by
the exact numbers \changed{given}
in Corollary~\ref{coro:levels_hyperplanes_convex}.

\subsection{Triangles, simplices, and piecewise linear functions}

Let $T = \{\triangle_1, \ldots, \triangle_n\}$ be a given set of $n$~triangles in~$\Real^3$.
For $1\leq k\leq n$,
the \emph{$k$-level} of the arrangement~$\arr(T)$ of triangles in~$T$
is defined to be the closure of the set of all points~$p$ on triangles in~$T$
such that the downward vertical ray from~$p$ meets exactly $k-1$ triangles.
Agarwal et al.~\cite{aacm-lalspt-98} proved that
the complexity of the $k$-level of~$\arr(T)$ is $O(k^{7/9}n^2\alpha(n/k))$
and Katoh and Tokuyama~\cite{kt-klcs-02} improved it to $O(k^{2/3} n^2)$.

To make this fit in our framework, for each $1\leq i\leq n$,
let $S_i$ be the set of four planes in~$\Real^3$
consisting of the plane containing~$\triangle_i\in T$
and three more planes through each side of~$\triangle_i$
that are almost vertical and go below~$\triangle_i$.
Regard each $1\leq i \leq n$ as a color from $K:=\{1,\ldots, n\}$,
and let $S:=\bigcup_{i\in K} S_i$ and $\kappa \colon S \to K$
such that $\kappa(s) = i$ if $s\in S_i$.
As above, let $\overline{E}_i$ be the upper envelope of planes in~$S_i$.
Observe then that the $k$-th level from \emph{below} of the arrangement 
$\overline{\arr} = \arr(\{\overline{E}_1, \ldots, \overline{E}_n\})$
coincides with the $k$-level of $\arr(T)$.
From the definition of~$\overline{\conf}(S, \kappa)$ as declared above,
notice that
the weight of each $c$-chromatic vertex~$v$ of~$\overline{\arr}$ is 
indeed $n-c-k$
if the downward vertical ray from~$v$ intersects exactly $k-1$ triangles in~$T$.
So, the weights and the levels are somehow in the reversed order in this case.

Hence, applying Corollary~\ref{coro:levels_hyperplanes},
we obtain the $O(k^2 n)$ bound for the $(\leq k)$-level from \emph{above} 
or, equivalently, for the $(\geq n-k)$-level from \emph{below} of~$\overline{\arr}$.
On the other hand, we can also obtain an upper bound on the $(\leq k)$-level of
the arrangement~$\arr(T)$ of triangles
by considering those vertices of~$\overline{\arr}$ in~$\overline{\conf}(S,\kappa)$ 
whose weights are at least $n-k$.
Furthermore, the same arguments are applied to
$(d-1)$-simplices in~$\Real^d$ for any constant $d\geq 2$ as follows.
\begin{theorem} \label{thm:levels_triangles}
 Let $T$ be a set of $n$ $(d-1)$-simplices in~$\Real^d$ for constant~$d\geq 2$,
 and $\arr(T)$ be their arrangement.
 For $1\leq k\leq n$,
 the number of vertices in the $(\leq k)$-level of~$\arr(T)$
 is $O(k n^{d-1} \alpha(n/k))$;
 the number of vertices in the $(\geq k)$-level of~$\arr(T)$
 is $O((n-k)^{\lceil d/2 \rceil} n^{\lfloor d/2 \rfloor})$.
\end{theorem}
\begin{proof}
Let us call a vertex of $\arr(T)$ \emph{$c$-chromatic}
if it appears as a $c$-chromatic vertex in~$\overline{\arr}$,
that is, the intersection of three planes from $c$ different sets~$S_i$.
Recall that $\overline{\conf}_{c,j}(S, \kappa)$ be
the set of $c$-chromatic weight-$j$ vertices of $\overline{\arr}$;
we have $v \in \overline{\conf}_{c,j}(S, \kappa)$
if and only if $v$ is a $c$-chromatic vertex of $\arr(T)$ such that
the downward vertical ray emanating from~$v$ intersects exactly $n-c-j$ triangles in~$T$.
Hence, for each $1\leq k\leq n$,
the $c$-chromatic vertices in the $(\leq k)$-level of $\arr(T)$
correspond to the $c$-chromatic colored configurations of weight at least $n-c-k+1$
in this setting.

Now, let $r$ be an integer parameter with $1\leq r\leq n$
and $R\subseteq K=\{1,\ldots, n\}$ be a random subset of $r$~colors.
For each $c \in \{1,2,3\}$,
we have 
 \[ \E[|\overline{\conf}_{c,r-c}(S_R, \kappa_R)|] 
     \geq \sum_{j=0}^{n-c}|\overline{\conf}_{c,j}(S,\kappa)| \binom{j}{r-c} \Big/ \binom{n}{r} \]
by Lemma~\ref{lem:CS_general_lb} (with $a=r-c$), on one hand.
On the other hand,
observe that $\bigcup_{c} \overline{\conf}_{c,r-c}(S_R, \kappa_R)$
consists of all vertices on the \emph{lower envelope} of $r$ upper envelopes
$\{\overline{E}_i\}_{i\in R}$ or, equivalently,
all vertices on the lower envelope of $r$~triangles in~$\{\triangle_i\}_{i\in R}$.
Since the complexity of the lower envelope of $r$~triangles in~$\Real^3$
is known as $O(r^2 \alpha(r))$~\cite{ps-ueplf-89,t-ntasapls-96},
we have 
\[ \E[|\overline{\conf}_{c,r-c}(S_R, \kappa_R)|] \leq 
    \sum_{b=1}^3 \E[|\overline{\conf}_{b,r-b}(S_R, \kappa_R)|] 
     = O(r^2 \alpha(r)),
\]
as $|S_i| = 4$ for every $i\in K$.

Fix $c\in \{2,3\}$ and set $r = \lfloor \frac{n}{k} \rfloor$.
From the above lower bound, 
we then obtain
\begin{align*}
 \E[|\overline{\conf}_{c,r-c}(S_R, \kappa_R)|]
  & \geq \sum_{j=n-c-k+1}^{n-c}|\overline{\conf}_{c,j}(S,\kappa)|\cdot\binom{j}{r-c} \Big/ \binom{n}{r}\\
  & = \sum_{i=0}^{k-1}|\overline{\conf}_{c,n-c-i}(S,\kappa)|\cdot \binom{n-c-i}{r-c} \Big/ \binom{n}{r}\\
  & \geq \left( \sum_{i=0}^{k-1}|\overline{\conf}_{c,n-c-i}(S,\kappa)|\right)
     \cdot \frac{r(r-1)\cdots(r-c+1)}{n(n-1)\cdots(n-c+1)}
     \cdot \left(\frac{c-1}{c}\right)^c
\end{align*}
if $k \leq \lfloor \frac{n}{c} \rfloor$
by the same derivation as in the proof of Theorem~\ref{thm:CS_general}.
Combining this with the above upper bound,
we get
\[ 
 \sum_{j=n-c-k+1}^{n-c}|\overline{\conf}_{c,j}(S,\kappa)| 
   = O\left(k^{c-2}\cdot n^2\cdot \alpha\left(\frac{n}{k}\right)\right)
\]
for $1\leq k \leq \lfloor \frac{m}{c} \rfloor$.
Note that the number of $1$-chromatic vertices in~$\overline{\arr}$ is $O(n)$ in total.
Therefore, the number of vertices in the $(\leq k)$-level of~$\arr(T)$ is bounded by
$O(k n^2 \alpha(n/k))$.
Finally, if $k > \lfloor \frac{n}{c} \rfloor$,
then we verify that $O(k n^2 \alpha(n/k)) = O(n^3)$,
which is asymptotically the same as the maximum possible number of 
vertices in $\overline{\arr}$ and in $\arr(T)$.
Hence, the claimed bound holds for any $1\leq k\leq n$.

The same approach can also be applied to the arrangement of
$(d-1)$-simplices in~$\Real^d$ for any constant $d\geq 2$.
It is known that the complexity of the upper envelope of $r$~simplices in~$\Real^d$
is bounded by $O(r^{d-1}\alpha(r))$~\cite{sa-dsstga-95,e-ueplf:tbnf-89,t-ntasapls-96}.
Hence, the first bound follows.

The second bound is implied by Corollary~\ref{coro:levels_hyperplanes}
since,
as discussed above,
the $(\geq k)$-level of~$\arr(T)$
corresponds to $(\leq n-k)$-level of $\overline{\arr}$ from above.
\end{proof}
Remark that, for $d=2$,
Theorem~\ref{thm:levels_triangles} implies the $O(kn\alpha(n/k))$ bound,
which is asymptotically the same as the known bound by Sharir~\cite[Theorem 1.2]{s-ksacs-91}
for line segments in~$\Plane$.

An analogous argument can also be applied to piecewise linear functions.
Let $F = \{f_1, \ldots, f_m\}$ be a collection of $(d-1)$-variate piecewise linear functions
that are fully or partially defined on a subset $D_i\subseteq \Real^{d-1}$,
consisting of one or more connected components bounded by linear faces.
Suppose that the domains $D_i$ are triangulated into $(d-1)$-simplices
and let $n$ denote the total number of those simplices.
Consider the arrangement $\arr(F)$ of the graphs of $m$ functions in~$F$,
and its $k$-level is defined analogously as above for the arrangement of triangles.
Observe that the vertices of $\arr(F)$ are colored configurations
by an analogous construction as done above.
Theorems~\ref{thm:CS_general} and~\ref{thm:CS_general_uniform} imply the following.
\begin{corollary}\label{coro:piecewise_linear}
 Given a set $F$ of $m$ $(d-1)$-variate piecewise linear functions 
 with a total of $n$ linear pieces as above,
 for $1\leq k \leq m$,
 the number of vertices in the $(\leq k)$-level of the arrangement of
 the graphs of those functions in~$F$ is
 $O(k m^{d-2} n^{d-1}\alpha(n/k))$; or 
 $O(k n^{d-1}\alpha(n/k))$ if the number of pieces of 
 each function in~$F$ is bounded by $\rho\cdot \frac{n}{m}$ for a constant~$\rho$.
 If $d=2$, then the bound is reduced to $O(k n \alpha(m/k))$.
\end{corollary}
\begin{proof}
Using the known upper bound on the lower envelope of 
simplices~\cite{sa-dsstga-95,e-ueplf:tbnf-89,t-ntasapls-96},
we apply Theorems~\ref{thm:CS_general} and~\ref{thm:CS_general_uniform}
with $T_0(n) = O(n^{d-1}\alpha(n))$.
The number of vertices in the $(\leq k)$-level for $1\leq k\leq m$ is thus 
bounded by
\[ 
 O\left(\frac{k^d}{m}\cdot \left(\frac{mn}{k}\right)^{d-1}\cdot \alpha\left(\frac{mn}{k}\right)\right)
 = O\left(k m^{d-2} n^{d-1} \alpha\left(\frac{n}{k}\right)\right).
\]
in general.
If the number of pieces of any two functions in~$F$ differ by a constant,
then the corresponding color assignment is almost uniform,
so we have the bound $O(k n^{d-1}\alpha(n/k))$. 

In case of $d=2$,
we have a better bound $O(n \alpha(r))$
on the complexity of the lower envelope of
any subset of $r$~functions in~$F$
by Har-Peled~\cite{h-mcl-99}.
So, the claimed bound follows from Theorem~\ref{thm:CS_general}.
\end{proof}

\subsection{Piecewise algebraic functions}

The results of Corollary~\ref{coro:piecewise_linear} 
are again extended to piecewise Jordan arcs 
and to piecewise algebraic functions.
In particular in~$\Real^2$,
Har-Peled~\cite{h-mcl-99} considered
the \changed{overlay}
of arrangements of Jordan arcs
and proved a general upper bound on
a single cell and many cells.
Theorem~\ref{thm:CS_general}, 
together with the results of~\cite{h-mcl-99}, 
we obtain the following,
extending the uncolored analog by Sharir~\cite[Theorem 1.3]{s-ksacs-91}
(see also Sharir and Agarwal~\cite[Corollary 5.18]{sa-dsstga-95}).

\begin{corollary} \label{coro:levels_Jordan}
 Let $S$ be a collection of $n$ $x$-monotone Jordan arcs,
 possibly being unbounded curves,
 such that any two of them intersect at most $t$ times,
 and $(S_1, S_2, \ldots, S_m)$ be a partition of~$S$ 
 into $m$~nonempty subsets.
 Let $E_i$ be the lower envelope of those in~$S_i$
 for $1\leq i\leq m$, and $\arr = \arr(\{E_1, \ldots, E_m\})$ be their arrangement.
 For $1\leq k\leq m$, let $C_{\leq k}$ be the number of vertices in the $(\leq k)$-level of $\arr$.
 \begin{itemize}
  \item In general, $C_{\leq k} = O(kn \cdot \beta_{t+2}(\frac{n}{k}))$,
  where $\beta_{t'}(n') := \lambda_{t'}(n') / n'$ and $\lambda_{t'}(n')$ denotes
  the maximum length of Davenport--Schinzel sequences of order~${t'}$ with $n'$ symbols.
  \item If $S$ consists of unbounded Jordan curves,
  then $C_{\leq k} = O(kn \cdot \beta_t(\frac{n}{k}))$.
  \item If Jordan arcs in~$S_i$ are disjoint for every~$i$, 
  so $E_i = \bigcup_{s\in S_i} s$,
  then $C_{\leq k} = O(kn \cdot \beta_{t+2}(\frac{m}{k}))$.
  \item If Jordan arcs in~$S_i$ are disjoint 
  and every vertical line intersects $E_i$ for every~$i$,
  that is, $E_i$ is the graph of a fully-defined function over~$\Real$,
 then  
 $C_{\leq k} = O(kn \cdot \beta_{t}(\frac{m}{k}))$.
 \end{itemize}  
\end{corollary}
\begin{proof}
The first two claims follow from Theorem~\ref{thm:CS_general}
with the upper bound on the lower envelope of $x$-monotone Jordan arcs 
or unbounded Jordan curves~\cite{sa-dsstga-95}.

We exploit the multicolor combination lemma by Har-Peled~\cite[Theorem 2.1]{h-mcl-99}.
Since each~$S_i$ is chosen such that $E_i$ is $x$-monotone,
each face of the arrangement~$\arr(S_i)$ of those Jordan arcs in~$S_i$
is linear to $|S_i|$.
Also, the arrangement $\arr = \arr(\{E_1, \ldots, E_m\})$ is
the overlay of $\arr(S_1), \ldots, \arr(S_m)$.
Hence, Theorem~2.1 of~\cite{h-mcl-99} implies
that the complexity of any single cell in the overlay arrangement~$\arr$
is $O(n\cdot \frac{\lambda_{t+2}(m)}{m}) = O(n \cdot \beta_{t+2}(m))$.
With this bound, Theorem~\ref{thm:CS_general} implies the third bound.

When $S$ consists of $x$-monotone unbounded Jordan curves
and $E_i$ is taken as the lower envelope of those in~$S_i$,
the complexity of a single cell in~$\arr$ is reduced 
to $O(n\cdot \frac{\lambda_{t}(m)}{m}) = O(n \cdot \beta_{t}(m))$~\cite[Lemma~2.3]{h-mcl-99},
so the fourth claim follows.
\end{proof}

Next, we consider
a collection of colored surface patches in~$\Real^d$ for constant $d\geq 2$.
Specifically,
let $S$ be a collection of $n$ algebraic surface patches in~$\Real^d$
that are graphs of a partially-defined $(d-1)$-variate algebraic functions.
It is known that the complexity of the lower envelope of $F$
is $O(n^{d-1+\epsilon})$ for any positive real~$\epsilon>0$
under some assumptions~\cite{hs-nblwtd-94,s-atublehd-94}.
(See also the book by Sharir and Agarwal~\cite[Chapter 7]{sa-dsstga-95}.)
Suppose $S$ is partitioned into $S_1, \ldots, S_m$
by any color assignment~$\kappa$ to $m$~colors.
Let $E_i$ for $1\leq i \leq m$ be the lower envelope of surface patches in~$S_i$,
and $\arr = \arr(\{E_1,\ldots, E_m\})$ be their arrangement.
For $1\leq k\leq m$, 
the $k$-level of $\arr$ is defined as above to be 
the closure of the set of points~$x$ on the surfaces~$E_i$
such that the downward vertical ray from~$x$ crosses exactly $k-1$ those surfaces.
Then, Theorems~\ref{thm:CS_general}--\ref{thm:CS_general_uniform}
and a derivation analogous to the proof of Theorem~\ref{thm:levels_triangles}
imply upper bounds on the $(\leq k)$-level of~$\arr$.
Hence, we conclude:
\begin{corollary} \label{coro:levels_patches}
 Let $F$ be a set of $m$ $(d-1)$-variate piecewise algebraic functions
 of maximum constant degrees with a total of $n$ algebraic pieces.
 Then, the number of vertices in the $(\leq k)$-level of the 
 arrangement $\arr(F)$ of the graphs of the functions in~$F$
 is bounded by $O(k^{1-\epsilon} m^{d-2} n^{d-1+\epsilon})$
 for any $\epsilon>0$.
 If the number of algebraic pieces in each function in~$F$ 
 is bounded by $\rho\cdot \frac{n}{m}$ for a constant~$\rho$,
 then the bound is reduced to $O(k^{1-\epsilon} n^{d-1+\epsilon})$.
\end{corollary}
This extends the known bounds for uncolored cases;
see Corollaries~7.8 and~7.18 in~\cite{sa-dsstga-95}.

\subsection{Convex polyhedra}

Another interesting structure that fits in the colorful Clarkson--Shor framework
is the arrangement of convex polyhedra.
Let $P_1, \ldots, P_m$ be $m$ given convex polyhedra, bounded or unbounded, in~$\Real^d$
with a total of $n$~facets, and $\arr = \arr(\{P_1,\ldots,P_m\})$ be their arrangement.
In this case, the \emph{depth} of any point \changed{$x\in\Real^d$} is often defined to be
the number of polyhedra~$P_i$ that contain $x$ in its interior.
We are interested in the number of vertices in~$\mathcal{A}$
whose depth is at most $k$.

We interpret this in our framework as follows.
For each facet $f$ of $P_i$, consider the open half-space, bounded by
the hyperplane spanning $f$, that \emph{avoids}~$P_i$.
Let $S$ be the set of all those $n$ open half-spaces
and $\conf(S)$ be the set of all vertices in the arrangement of these
bounding hyperplanes.
We say that a half-space $s \in S$ is in conflict with a vertex $v\in \conf(S)$ 
if $v$ is contained in~$s$.
Now, we consider the color assignment $\kappa \colon S \to K =\{1,\ldots, m\}$,
according to its original polyhedron~$P_i$ for~$i\in K$.
Observe that
the set $\conf(S,\kappa)$ of colored configurations, induced by~$\conf(S)$ with respect to~$\kappa$,
consists of all vertices in $\mathcal{A}$,
and $\conf_{c,j}(S,\kappa)$ is the set of all $c$-chromatic vertices of depth~$m-c-j$.
That is, in this case, the depths are ordered in the reversed way to the weights.

For $d=2$, Aronov and Sharir~\cite{as-cecpp-97} proved that the complexity
of the common exterior of polygons $P_1, \ldots, P_m$
or, equivalently, 
the number of vertices of depth~$0$ in~$\mathcal{A}$ is bounded by $O(n\alpha(m) + m^2)$
in general, and $O(n \alpha(m))$ if the common exterior is connected.
With these upper bounds, Lemma~\ref{lem:CS_general_lb} yields the following
through a similar derivation as done in the proofs of
Theorems~\ref{thm:CS_general} and~\ref{thm:levels_triangles}.

\begin{theorem} \label{thm:polygons_depth}
 Given $m$ convex polygons of a total of $n$ sides in~$\Plane$,
 let $\arr$ be their arrangement.
 For~$0\leq k\leq m-1$,
 the number of vertices of depth at most~$k$ in~$\arr$
 is $O((k+1)n\cdot \alpha(\frac{m}{k+1})+m^2)$.
 If the common exterior of any subset of the $m$ polygons is connected,
 then the bound is reduced to $O((k+1)n\cdot \alpha(\frac{m}{k+1}))$.
\end{theorem}
\begin{proof}
Recall that $c$-chromatic vertices in~$\mathcal{A}$ of depth~$j$
are those configurations of weight $m-c-j$ in~$\conf_{c,m-c-j}(S,\kappa)$.
Hence, the claimed bounds can be shown by a similar derivation
as in the proof of Theorems~\ref{thm:CS_general} and~\ref{thm:levels_triangles}.

Let $1\leq r\leq m$ be an integer parameter and $R\subseteq K$ be a random set of $r$~colors.
In general,  we have
\begin{align*}
 \E[|\conf_{c,r-c}(S_R, \kappa_R)|] &= O(\E[|S_R|]\cdot \alpha(r) + r^2)\\
   & = O\left(\left(\sum_{R'\subseteq K, |R'|=r} |S_{R'}| \Big/ \binom{m}{r} \right) 
     \cdot \alpha(r) + r^2\right) \\
   & = O\left(\left(\binom{m-1}{r-1} n \Big/ \binom{m}{r}\right) \cdot \alpha(r) + r^2 \right)\\
   & = O\left( \frac{r}{m} n \cdot \alpha(r) + r^2 \right)
\end{align*}
by Aronov and Sharir~\cite{as-cecpp-97} 
and Lemma~\ref{lem:sums},
on one hand.
On the other hand, setting $r = \lfloor \frac{m}{k+1} \rfloor$,
Lemma~\ref{lem:CS_general_lb} (with $c=2$ and $a=r-2$) implies
\begin{align*}
 \E[|\conf_{2,r-2}(S_R, \kappa_R)|] & \geq \sum_{j=0}^{m-2} |\conf_{2,j}(S,\kappa)|
       \binom{j}{r-2} \Big/ \binom{m}{r}\\
    & \geq \sum_{j=m-2-k}^{m-2} |\conf_{2,j}(S,\kappa)|
       \binom{j}{r-2} \Big/ \binom{m}{r}\\
    & = \sum_{i=0}^{k} |\conf_{2,m-2-i}(S,\kappa)|
       \binom{m-2-i}{r-2} \Big/ \binom{m}{r}\\
    & \geq \left(\sum_{i=0}^{k} |\conf_{2,m-2-i}(S,\kappa)|\right)
        \cdot \frac{r(r-1)}{m(m-1)} \cdot \left(\frac{1}{2}\right)^2
\end{align*}
if $k \leq \lfloor \frac{m}{2} \rfloor - 1$.

Combining the two inequalities results in the first bound
 \[
  \sum_{j=m-2-k}^{m-2} |\conf_{2,j}(S,\kappa)|
 = O\left( (k+1)n\alpha\left(\frac{m}{k+1}\right) + m^2 \right),
 \]
for $0\leq k\leq \lfloor \frac{m}{2} \rfloor - 1$,
and the number of $1$-chromatic vertices is subsumed by this bound.
One can easily check the same bound holds for $\lfloor \frac{m}{2} \rfloor \leq k \leq m-1$,
since the total number of vertices of the arrangement~$\arr$ is bounded by $O(mn + m^2)$.

The second one can also be derived in a similar way
with the upper bound
 \[ \sum_{c=1}^2 \E[|\conf_{c,r-c}(S_R, \kappa_R)|] = O(\E[|S_R|]\cdot \alpha(r))
   = O\left(\frac{r}{m} n \cdot \alpha(r)\right) \]
if the common exterior of any subset of the $m$ polygons is connected,
as shown by Aronov and Sharir~\cite{as-cecpp-97}.
\end{proof}
Remark that the second bound of Theorem~\ref{thm:polygons_depth} holds
even for simple polygons.
Har-Peled~\cite[Lemma 2.8]{h-mcl-99} proved an upper bound $O(n \alpha(m))$
on the complexity of a single cell in the arrangement of $m$ simple polygons
with $n$ total sides.
Thus, if the common exterior of any subset of the $m$ simple polygons is connected,
the number of vertices of depth~$0$ in their arrangement is $O(n \alpha(m))$,
hence the same upper bound $O((k+1)n\alpha(m/(k+1)))$ is derived for
the number of vertices of depth at most~$k$.

Similarly, for $d=3$, we conclude the following
based on the results of Aronov et al.~\cite{ast-ucptd-97} and Ezra and Sharir~\cite{es-scacp-07}.
\begin{corollary}
 Given $m$ convex polyhedra, bounded or unbounded, of a total of $n$ faces in~$\Real^3$
 and an integer $0\leq k \leq m-1$,
 the number of vertices in their arrangement of depth at most~$k$
 is $O((k+1)mn\log(\frac{m}{k+1})+m^3)$.
 If the common exterior of any subset of the $m$ polyhedra is connected,
 then the bound becomes $O((k+1)^{1-\epsilon}m^{1+\epsilon}n)$ for any $\epsilon>0$.
\end{corollary}
Note that the first bound has been mentioned by Aronov et al.~\cite[Theorem~1.7]{ast-ucptd-97}.

\section{Concluding Remarks} \label{sec:conclusion}
We finish the paper with some remarks and further questions.

The colorful Clarkson--Shor framework provides
a systematic scheme to handle families of configurations or geometric ranges
defined by objects
of
non-constant complexity.
We showed its \changed{application}
to color Voronoi diagrams,
colored $j$-facets, and arrangements of various curves and surfaces
of non-constant complexity.
Our general upper bounds, shown in \changed{Theorems~\ref{thm:CS_general}} and~\ref{thm:CS_general_uniform},
can be applied once
we have obtained any
function~$T_0$ that upper bounds
the number of weight-$0$ uncolored configurations.
While almost the same bounds as in
the uncolored case hold
when $T_0$ is near-linear,
it seems hard to avoid the extra term in the number~$m$ of colors
in general for any color assignment.
Is it possible to obtain the original Clarkson--Shor bound $O((k+1)^c \cdot T_0(n/(k+1)))$,
or similar bound,
on the number of $c$-chromatic weight-$(\leq k)$ colored configurations
under a reasonable requirement on~$T_0$?

In this paper,
we introduced the higher-order color Voronoi diagrams~$\CVD_k(S)$ and $\mCVD_k(S)$
with distance-to-site functions~$\delta_s$ for $s\in S$.
Our combinatorial results
rely on the general position \changed{of the}
functions~$\delta_s$
and conditions~\textbf{V1}--\textbf{V3} on numbers of vertices and unbounded edges
in ordinary nearest and farthest-site Voronoi diagrams.
Can we drop condition~\textbf{V3} to obtain
the same
upper bound~$4k(n-k)-2n$ on the total number of vertices
in $\CVD_k(S)$ and $\mCVD_k(S)$?
In Section~\ref{sec:kcvd_general},
we showed that this can be done for any convex distance function
by a limit argument.

We
presented an iterative approach to compute \changed{order-$k$} color Voronoi
diagrams \changed{under general distance functions that satisfy the
  conditions of 
  abstract Voronoi diagrams;
  and an additional
  condition~\textbf{V3$^\prime$} for the  maximal order-$k$ color counterpart.}
Can one achieve a faster algorithm that computes a \changed{specific}
order-$k$ color Voronoi diagram 
under the Euclidean metric?
Or, can the approach using nondeterminism by Chan et al.~\cite{ccz-oahovdp:usn-24}
be extended to color Voronoi diagrams?

\section*{Acknowledgment} 
The authors would like to thank Otfried Cheong, Christian Knauer, and Fabian Stehn 
for valuable discussions and comments,
in particular, at the beginning of this work.
The research by the first author has been done
partly during his visits to Universit\"{a}t Bayreuth, Bayreuth, Germany
and to Universit\`{a} della Svizzera italiana, Lugano, Switzerland.




\bibliography{kcvd_socg25_arxiv}

\end{document}